\documentclass{article}

\usepackage{arxiv}

\usepackage[utf8]{inputenc} 
\usepackage[T1]{fontenc}    
\usepackage{url}            
\usepackage{booktabs}       
\usepackage{amsfonts,amsmath,amssymb,amsthm}       
\usepackage{nicefrac}       
\usepackage{microtype}      
\usepackage{lipsum}
\usepackage{graphicx}
\graphicspath{ {./images/} }
\usepackage{multirow}
\usepackage{algorithm}
\usepackage[noend]{algorithmic} 
\usepackage[pagebackref]{hyperref}
\usepackage{color}
\usepackage{paralist}
\usepackage{soul}

\newboolean{isDraft}
\setboolean{isDraft}{false}
\ifthenelse{\boolean{isDraft}}{
\usepackage{fullpage}
\usepackage{setspace}
\doublespacing
\usepackage[right]{lineno}

\linenumbers
}{
\usepackage{fullpage}
}

\theoremstyle{plain}
\newtheorem{theorem}{Theorem}

\newtheorem{lemma}         [theorem]{Lemma}

\newtheorem{claim}         [theorem]{Claim}

\newboolean{short}
\setboolean{short}{true}

\newcommand{\shortOnly}[1]{\ifthenelse{\boolean{short}}{#1}{}}
\newcommand{\onlyShort}[1]{\ifthenelse{\boolean{short}}{#1}{}}
\newcommand{\longOnly}[1]{\ifthenelse{\boolean{short}}{}{#1}}
\newcommand{\onlyLong}[1]{\ifthenelse{\boolean{short}}{}{#1}}

\newcommand{\G}{\mathcal{G}}

\newcommand{\C}{\mathcal{C}}
\newcommand{\A}{\mathcal{A}}

\newcommand{\V}{\mathcal{V}}

\renewcommand{\L}{\mathcal{L}}

\newcommand{\bgt}{\texttt{budget}}
\newcommand{\ch}{{\texttt{C}}}
\renewcommand{\t}{{\texttt{T}}}

\newcommand{\ccg}{{\textsc{CML}}}

\newcommand{\cml}{{\textsc{CML}}}

\newcommand{\plog}{\text{polylog}}
\newcommand{\poly}{\text{polynomial}}
\newcommand{\Pe}{P2P }

\newcommand{\E}[1]{{\rm I\kern-.3em E}\left [#1 \right ]}

\renewcommand{\r}{\rho}

\renewcommand{\c}{\gamma}

\newcommand{\red}[1]{{\textcolor{red}{#1}}}
\newcommand{\blue}[1]{\textcolor{blue}{#1}}

\newcommand{\magenta}[1]{\textcolor{magenta}{#1}}

\def\ShowComment{true}
\ifdefined\ShowComment
\def\john#1{\marginpar{$\leftarrow$\fbox{J}} \magenta{ $\bigstar$}\footnote{$\Rightarrow$~{\sf #1 \magenta{--John}}}}
\def\sumathi#1{\marginpar{$\leftarrow$\fbox{S}}\blue{ $\bigstar$}\footnote{$\Rightarrow$~{\sf #1 \blue{--Sumathi}}}}
\def\reviewer1#1{\marginpar{$\leftarrow$\fbox{R1}}\red{ $\bigstar$}\footnote{$\Rightarrow$~{\sf #1 \red{--reviewer1}}}}

\else
\def\john#1{}
\def\sumathi#1{}
\def\reviewer1#1{}

\fi

\title{Spartan: Sparse Robust Addressable Networks\thanks{tThis research was supported in part by an Extra-Mural Research Grant (file number EMR/2016/003016) funded by the Science and Engineering Research Board, Department of Science and Technology (SERB), Government of India. The first author is also supported by an SERB MATRICS project (file number MTR/2018/001198) and an Indo-German joint project supported by DST and DAAD (INT/FRG/DAAD/P-25/2018).} \thanks{Preliminary results of this paper appeared in the Proceedings of the 32nd IEEE International Parallel \& 
		 Distributed Processing Symposium (IPDPS), 2018~\cite{AS18}}}

\author{
 John Augustine \\
  Computer Science and Engineering\\
    Indian Institute of Technology Madras\\
    Chennai, India \\
  \texttt{augustine@cse.iitm.ac.in} \\
   \And
Sumathi Sivasubramaniam\\
Computer Science and Engineering \\
 Indian Institute of Technology \\
 Chennai, India\\
  \texttt{sumathi@cse.iitm.ac.in} \\
  \And
}

\begin{document}
\maketitle
\begin{abstract}
A Peer-to-Peer (P2P) network is  a dynamic collection of nodes that connect with each other via  virtual overlay links built upon an underlying network (usually, the Internet). P2P networks are  highly dynamic and can experience very heavy churn, i.e., a large number of nodes join/leave the network continuously. Thus, building and maintaining a stable overlay network  is  an important problem that has been studied extensively for two decades.

In this paper, we present our \Pe overlay network called Sparse Robust Addressable Network (Spartan). Spartan can be quickly and efficiently built in a fully distributed fashion within $O(\log n)$ rounds. Furthermore, the Spartan overlay structure  can be  maintained, again, in a  fully distributed manner despite adversarially controlled churn (i.e., nodes joining and leaving) and significant variation in the number of nodes. Moreover, new nodes can join a committee within $O(1)$ rounds and leaving nodes can leave without any notice.

The number of nodes in the network  lies in $[n, fn]$ for any fixed $f\ge 1$. Up to $\epsilon n$ nodes  (for some small but fixed $\epsilon > 0$) can be adversarially added/deleted within {\em any} period of $P$ rounds for some $P \in O(\log \log  n)$. Despite such uncertainty in the network, Spartan maintains $\Theta(n/\log n)$ committees that are stable and addressable collections of $\Theta(\log n)$ nodes each for $\poly(n)$ rounds with high probability. 

Spartan's committees are also capable of performing sustained computation and passing messages between each other. Thus, any protocol designed for static networks can be simulated on Spartan with minimal overhead. This makes Spartan an ideal platform for developing applications. We experimentally show that Spartan will remain robust as long as each committee, on average, contains 24 nodes for networks of size up to $10240$. 
\end{abstract}



\section{Introduction}

Peer-to-Peer (P2P) networks have become increasingly popular over the years, with applications such as Bitcoin,  Bittorrent, Cloudmark, etc. becoming  widespread in use. Such distributed and decentralized methods of networking facilitate resource  sharing and decentralized applications that are resilient and scalable.  In addition to the typical P2P context where end-users clients are connected to form the overlay, P2P principles have also influenced dedicated overlay infrastructure provided by commercial enterprises like Akamai~\cite{:2014aa}. 

\Pe networks are typically virtual networks --- called {\em overlay} networks ---  built on top of an underlying (often physical) network --- called the \textit{underlay}. A simple schematic is shown in Figure~\ref{fig:intro-schematic}. The underlay network is typically the Internet and only a relatively small (and dynamically changing) fraction of the nodes in the underlay network may choose to participate in the overlay P2P network. The underlay allows us to create  virtual edges in the network (e.g., TCP connections) that can serve as the basis for virtual overlay edges. In our work, our goal is to design a semi-structured overlay that provides stable addressable locations also called a distributed hash table (DHT) that facilitates storage and retrieval of data items. 

\begin{figure}[hptb]
\begin{center}
\includegraphics[scale=0.60, clip=true,trim=10 185 100 120]{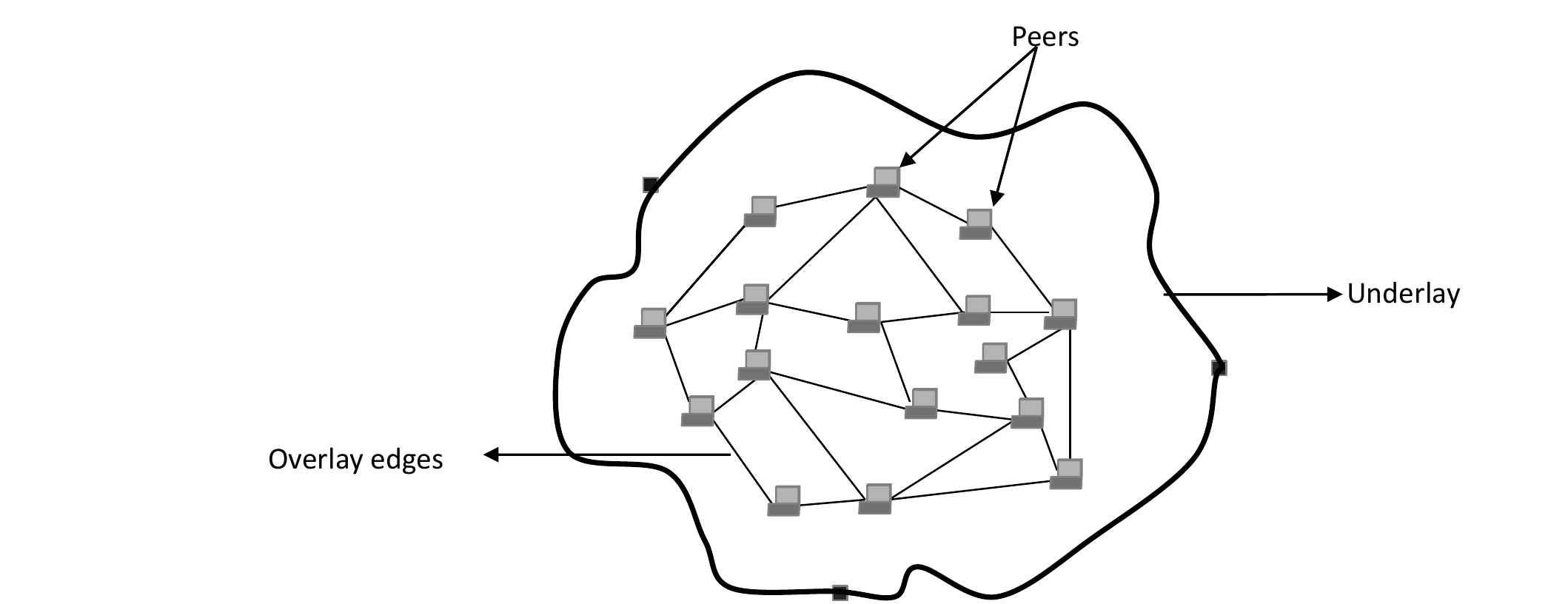}
  \end{center}
  \caption{General Structure of a \Pe network}\label{fig:intro-schematic}
\end{figure}

 
The challenge lies in the fact that \Pe networks tend to be highly dynamic and experience heavy churn, i.e., a large number of peer nodes join and leave the network at a rapid pace \cite{FPJKA07, GSG02, SMKKB01}.  In fact, real world studies of \Pe networks \cite{SR06} show that as much as fifty percent of a network can be replaced within a hour. This makes searching and indexing of data items notoriously difficult  because unpredictable churn can destroy stored data items. 
 
The ability to store, index, search, and retrieve data items addressed by keys is crucial in many P2P applications like BitTorrent, Bitcoin, and a host of other DHT implementations. Quite naturally, there has been a significant amount of work in building overlay structures that allow for efficient look-up of data items in \Pe networks.  Many overlay structures like Chord~\cite{SMKKB01}, CAN~\cite{RFHKS01}, Pastry~\cite{RD01}, Tapestry~\cite{ZKJ01}, and  Skip Net~\cite{HJSTW03} have been proposed.   However one problem that is persistent in all of these works is the lack of mathematical guarantee against heavy churn, i.e., up to $\Omega(n/\plog(n))$ nodes joining and leaving in every round. We refer to \cite{LCPSL05, M15} as good surveys of available research in this topic.

Recent works that do provide guarantees in the face of heavy churn, such as~\cite{AMMPRU13} make assumptions about the network topology  such as the network having good expansion properties. Research has shown that expander topologies may provide the right framework for solving several different problems in \Pe networks such as consensus~\cite{APRU12}, leader election~\cite{AKS15}, and even storage and search~\cite{AMMPRU13} for  a single data item. The good news is that we now have ways to build and maintain overlays with good expansion properties despite heavy churn~\cite{APRRU15}. 
However,  such expander  topologies  tend to be unstructured and  prove to be inefficient for  storing, searching, and retrieving data items, which is a key requirement for DHT applications.

\subsection{Model Characteristics} 
We present a robust overlay network called {\em Sparse Robust Addressable Network (Spartan)} that allows for the storage and look up of data items in \Pe networks. The number of nodes may vary over time, but for ease of exposition, we assume that the number of nodes is in the range $[n, fn]$ for some constant $f \ge 1$. As the name implies, it is a sparse network with maximum degree bounded by $O(\log n)$.  Our model brings together a wide range of modeling ideas found in several prior works \cite{FS02,KSW10,AS09,AMMPRU13,APRRU15,DGS16} that are quite disparate with respect to their modeling approaches. Our hope is that our model will be a step towards a unified  approach to modeling this complex problem. We will highlight a few salient model characteristics here; a detailed description is provided in Section~\ref{sec:model}. 

Our network is provided by an almost adaptive adversary that knows the network and the details of Spartan up to a small $O(\log \log n)$ steps prior to the current round. This adversary moreover can design the network with high levels of dynamism, some aspects of which are heretofore unseen in the literature. First and foremost, the adversary can design the network with {\em heavy} and {\em erratic churn}, whereby, it is allowed to add/delete nodes  at the rate of up to  $\epsilon n$ adds/deletes within any time period $P \in O(\log\log n)$ for some fixed $\epsilon >0$. Such a  heavy churn rate means that the adversary  can replace all the nodes  in $O(\log \log n)$ rounds. To make matters worse, the churn can be erratic in the sense that there can be sudden spikes with up to $\epsilon n$ nodes added or deleted within $O(1)$ rounds (as long as it is preceded and succeeded by relative calm). Such heavy, erratic, and adversarial  nature of churn makes the maintenance of a robust overlay network a difficult challenge to overcome.  

Moreover, the degree of each node is at most $O(\log n)$. Each node that enters the network (after the first $\Theta(\log n)$ rounds called the bootstrap phase) is adversarially connected to one pre-existing node. To maintain a connected overlay in this setting will require a nuanced approach. To see the challenge, consider the following scenario in which there is a pre-existing node $v$ and a new node $u_0$  is connected to $v$  in some round $r$. In round $r+1$, $u_0$ is also a pre-existing node. So, $u_1$  enters in round $r+1$ and is connected to $u_0$, while $v$ is churned out. This continues with a new node $u_i$ in each round $r+i$ connected to $u_{i-1}$, while $u_{i-1}$ is churned out. In this chain, if a single node somehow fails to get connected properly to the overlay network, then all subsequent nodes will continue to be disconnected. In fact, in our model, up to $O(\log n)$ new nodes can be simultaneously connected to the same  pre-existing node. This means that a disconnected node could affect an exponential growing number of nodes as the rounds progress.

\subsection{Our Results and Technical Contributions}

Spartan is a framework that allows building and maintaining P2P overlay networks that are robust against heavy adversarial churn. 
We begin with the  Spartan architecture (see Section~\ref{sec:spartan}) and then describe procedures to build and maintain Spartan (see Section~\ref{sec:warmup} and Section~\ref{sec:maintain}).   With high probability (whp)\footnote{We say that an event occurs with high probability if, under suitable choice of constants used by algorithms, the probability of its occurrence is at least $1-\frac{1}{n^c}$, for any fixed  $c>0$.}, Spartan possesses the following features.

\begin{compactenum}
\item It provides  $\Theta(\frac{n}{\log n})$ virtually addressable locations (called {\em committees}). These  committees are collections of $\Theta(\log n)$ nodes connected together completely by overlay edges. See line number~\ref{lno:join} of Algorithm~\ref{alg:warmup-overview}, Lemma~\ref{lem:comsize} and Lemma~\ref{lem:committes} for more details.
\item These committees are in turn connected together by means of {\em logical edges}. Each logical edge between two committees is a collection of complete bipartite overlay edges connecting nodes from either committees. The committees and the logical edges put together form a {\em logical network}. See line numbers~\ref{lno:formccg} and  \ref{lno:formedges} of Algorithm~\ref{alg:warmup-overview}.
\item For concreteness, we have constructed the logical network to be a butterfly network with the committees as butterfly nodes and logical edges as butterfly edges. This provides us with the ability to  address each committee by its location in the butterfly structure. See Lemma~\ref{lem:butterfly}.
\item Every node in Spartan is within a distance of $O(\log n)$ of every committee in the network (as seen in lemma~\ref{lem:butterfly}). Moreover, there exists a simple bit-fixing algorithm by which it is easy to reach any required committee in $O(\log n)$ rounds (see Chapter 4 in~\cite{MU05}). 
\item  We show how to build and maintain Spartan so that it is provably robust against  heavy and erratic adversarial churn described earlier. In particular, we prove that the Spartan network survives (whp) for a number of rounds that is a large polynomial in $n$. See Section~\ref{sec:maintain} and Theorem~\ref{thm:main} for a full explanation. 
\item We show that an algorithm designed under a suitably defined CONGEST model of distributed computing for static networks can be efficiently simulated by the committees of Spartan. Therefore a host of applications and algorithms developed under the assumption that the network is static can be extended to dynamic networks. See Section~\ref{sec:simulation} for a full description.
\end{compactenum}


The key to Spartan's robustness lies in the strength of  committees, which are collections of $\Theta(\log n)$ nodes that work together to ensure persistence of activities (like storing, searching, and retrieving data required for building stable applications) despite churn. These committees are built to be self sustaining. Our procedures ensure that the committees stay  robust in the face of heavy and erratic adversarial churn.  Consequently, the network as a whole is robust. The notion of committees is not new. In fact, variations of our notion of committees have been employed in building overlays for over 15 years (for example, see Fiat and Saia~\cite{FS02}). Over the years, we have seen quite a few variations spanning both deterministic~\cite{KSW10} and randomized overlay networks~\cite{AMMPRU13,DGS16}. However, our design allows us to provide a  DHT solution that can be constructed quickly. Moreover, to the best of our knowledge, ours is the first known DHT construction that is provably capable of handling such heavy churn. Most importantly, we believe the conceptual core is very simple and easy to implement and provides a lot of flexibility for fine-tuning to meet specific requirements.

In the interest of clarity in exposition, we have not optimized our results for constants and $\log n$ factors. To avoid technicalities arising with ceilings and floors, we have  assumed  that  $\log n$ and $\log \log n$ are integral.


\subsection{Survey of Related Works} 
The notion of overlay networks sprang out of a need for bringing flexibility in operating a distributed system despite an underlying  network  infrastructure that we cannot modify. Viewed from this broad perspective, overlay networks have been studied since the 80's, for example in~\cite{KLMR89} where the focus was on developing new network functionalities over a fixed interconnection network. Since the late 90's, there has been an explosion of P2P ideas and technologies mostly motivated by the need for fully decentralized P2P networks. Ideas like consistent hashing~\cite{Karger:1997:CHR:258533.258660} and web caching~\cite{Plaxton:1997:ANC:258492.258523,Korupolu:1999:PAH:314500.314880}, and P2P tools like Napster and Gnutella served as harbingers to the developments we have seen in the last two decades.   
These and other early developments led to significant research in  P2P systems motivated mainly by the need for distributed (and often completely decentralized) churn-resistant distributed hash table (DHT) platforms for storing, searching, and retrieving resources that are in essence {\tt <key, value>} pairs.    

To elaborate a bit, research in overlay networks and P2P has since then moved along two parallel (but not necessarily disjoint) strands of ideas: structured and unstructured P2P networks. 
In structured P2P networks, nodes (or rather their IDs)  and the keys share the same hash space. Nodes are hashed into a location in the network based on their ID and take up responsibility for resources that hash into that same space. This will allow a node seeking some resource identified by a key to hash the key and approach the node at that hashed location for the required resource. The big advantage with this approach is that  there is now a clean and efficient algorithm to admit new nodes, store resources, and discover them when needed. The down side to this approach is that the nodes will usually have to be arranged in a somewhat predictable and rigid manner. This often makes such networks vulnerable to heavy churn and security breaches. 

Structured P2P concept is perhaps best illustrated by Chord~\cite{SMKKB01} developed by Stoica et al. Nodes in Chord are arranged in the form of a ring in which the hash values of IDs are in sorted order. Resources are then stored in the node with the closest matching hash value. Pastry~\cite{RD01} and Tapestry~\cite{ZKJ01} are a couple of other similar P2P overlays. They both in fact employ prefix routing~\cite{Plaxton:1997:ANC:258492.258523} presented by Plaxton et al. CAN (Content Addressable Networks) introduced by Ratnasamy et al.~\cite{RFHKS01} around the same time arranges the nodes in a $d$-dimensional space. If the space is filled up appropriately, routing is robust (with several alternative paths) and efficient; resources are on average within a distance $O(dn^{1/d})$. However, CAN is sometimes prone to failures due to partitioning of nodes.
Viceroy~\cite{MNR02} described by Malkhi et al. is another  overlay network developed early in the history of P2P systems. They were early adopters of a topology very similar to the butterfly which we  use. 
Another topology that has gained quite a bit of attention~\cite{NW03,LKRG03,FG06,RSS11} is de Bruijn graphs~\cite{B46}, which provide very short average distances between nodes. See Loguinov et al.~\cite{LKRG03} for a detailed study of the benefits of de Bruijn graphs.

In the alternative idea of unstructured overlay networks, nodes connect with each other randomly with guarantees of good expansion. This means that we can employ flooding~\cite{APRU12} and random walks based techniques~\cite{GMS04,APR13,AMMPRU13,APRRU15,APR15} for searching and restructuring. Although these are significantly more inefficient compared to structured networks, these unstructured topologies are significantly more resilient to the effects of churn. 

Although churn is a very important factor in designing P2P overlays, it has received relatively low attention especially from theoretical researchers. As we mentioned earlier, Fiat and Saia~\cite{FS02} provided a P2P overlay design very similar to ours, but their analysis is limited to a one time failure (censorship in particular) rather than a more persistent notion (the way churn operates in reality). Their similarity to our work is two-fold in that they employ a notion of ``supernodes" which is similar to what we call committees, and furthermore their supernodes are also strung together in the form of a butterfly. 

Liben-Nowell et al.~\cite{LBK02} were some of the earliest to recognize the need for a formal investigation of such persistent churn, and, towards this goal, provided a characterization of the rate of churn in terms of what they call the half life time period, which is defined as the time for a network to either double or halve in size (i.e., number of nodes.). Awerbuch and Scheideler~\cite{AS04} proposed Hyperring, which is a distributed dynamic deterministic data structure capable of performing concurrent inserts, deletes and searches, all in $\plog(n)$ time.  Similar randomized distributed data structures like Skip Graphs~\cite{AS07} and Skipnet~\cite{HJSTW03} have also been presented in literature, but the rates at which they can handle churn is limited. Jacobs and Pandurangan~\cite{JP13} have presented a DHT capable of handling churn with insertions and deletions incurring at most $O(\log n)$ overhead, but their claims are limited to stochastic insertions and deletions. An interesting deterministic P2P overlay network was proposed by Kuhn et al.~\cite{KSW10}. They illustrate their approach via hypercube and pancake topologies. Each ``node" in these topologies is actually a collection of peers and in this sense akin to our idea of committees. Their overlay can withstand a churn of up to $O(\log n)$ nodes joining and leaving concurrently in every round. Moreover, their adversary is deterministic and therefore capable of observing the weakest spot in the network and targeting that spot for churn. 

Recently, there has been a flurry of works that investigate P2P systems that experience persistent heavy churn, which we define as a P2P system in which up to some $\Omega(n/\plog(n))$ nodes can join and leave the network in every round. In \cite{APRU12}, we showed that an unstructured overlay network -- modeled as a network in which expansion is maintained in every round -- is amenable to solving the agreement problem despite churn that is linear in the size of the network.  In every round, up to $\epsilon n$ nodes can join and leave the network for some small constant $\epsilon>0$. We also extended some of these ideas to Byzantine agreement~\cite{APR13} and  leader election problems~\cite{APR15,AKS15}. We also showed in \cite{AMMPRU13} that we can store and retrieve a data item despite near linear number of nodes (i.e., up to $O(n/\plog(n))$ nodes) churn in every round.  In~\cite{AMMPRU13}, we  employed the notion of committees, but  indexing the item requires the participation of $\Theta(\sqrt{n} ~\plog(n))$ nodes. Our current work also employs this notion of committees, but without the indexing overhead.
More recently, we showed in~\cite{APRRU15} that expansion (in an almost everywhere sense) can be guaranteed despite heavy churn. Most of these works employ an oblivious adversary, i.e., one that is aware of the protocols and can design churn in the worst possible manner, but it must do so without knowing the exact random bits employed by the protocols.

Recently and independently Drees et al.~\cite{DGS16} presented a P2P overlay network that can tolerate churn to the extent that all the nodes in the network can be replaced within $O(\log \log n)$ rounds. Moreover, this is achieved by an $\Omega(\log \log n)$-late omniscient adversary (formally defined in Section~\ref{sec:model}) that is oblivious only to the random bits employed in the most recent $\Theta(\log \log n)$ rounds. There are a few crucial differences between our models and our techniques are entirely different. In their model, the adversary must forewarn node insertions and departures at least $\Theta(\log \log n)$ rounds before they actually leave, whereas, in our case, the nodes can leave abruptly. Thus, in~\cite{DGS16}, the challenge is limited to quickly rebuilding the network every $\Theta(\log \log n)$ rounds using only the stable nodes that are guaranteed to remain in the network.  Moreover, they do not show how their structure can be constructed. We show how Spartan can be constructed  efficiently (in just $O(\log n)$ rounds) and general algorithms designed to operate on stable networks can be effectively simulated on Spartan.

\paragraph*{Organization of the paper}
We begin with a formal description of the model in Section~\ref{sec:model}.
We then present an overview of the Spartan overlay structure  in Section~\ref{sec:spartan}.  The overall design is described in Section~\ref{sec:design}. Then, in Section~\ref{sec:simulation}, we show how we can use Spartan for stable distributed computation, in particular, for implementing stable distributed hash tables.
Then, in Section~\ref{sec:warmup}, we present protocols to construct  Spartan.  The maintenance protocols -- including the crucial random walks based sampling procedure -- are described in Section~\ref{sec:maintain}.  Finally, we discuss   a variety of extensions in Section~\ref{sec:conc} illustrating Spartan's flexibility and  end with some concluding remarks.

\section{Network Model and Problem Statement} \label{sec:model}

 We consider a synchronous system in which time is measured in rounds.  An adversary presents a  dynamically changing sequence of sets of network nodes $\V = (V_1, V_2, \ldots)$ where $V_r \subset U$, for some universe of nodes $U$ and $r \ge 1$, denotes the set of nodes present in the network during round $r$.  For simplicity in exposition, we require the number of nodes  $|V_r|$ at any round to lie in $[n, fn]$ for some fixed $f\ge 1$; this limitation can be circumvented as discussed in Section~\ref{sec:conc}. Each node in $U$ is assumed to have a unique ID that can be represented in $O(\log n)$ bits.  For the first $B \in \Theta(\log n)$ rounds called the {\em bootstrap phase}, $\V$ is stable; i.e., for $1 \le r < B$, $V_r = V_{r+1}$. Subsequently, the network is said to be in its {\em maintenance phase} during which $\V$ can experience churn in the sense that  a large number of nodes (specified shortly) may join and leave dynamically at each time step. 

The {\em churn rate}  models the level of churn that the adversary can inflict on $\V$. It is specified by a pair $(\ch, \t)$ which implies that within {\em any} range of $\t$ consecutive rounds, at most $\ch$ nodes can be added and (a possibly different number of) at most $\ch$ nodes can be deleted.   Formally, $\forall r\ge 1$ and $1\le i\le \t$, the adversary must ensure that   $|V_{r+i} \setminus V_r| \le \ch$ and $|V_r \setminus V_{r+i}| \le \ch$.  The churn rate we consider is $(\epsilon n, \Theta(\log \log n))$ for a suitably small but fixed  $\epsilon>0$.

Any node $u$  can communicate with another  node $v$ and/or form an edge between them (with $v$'s concurrence) provided $u$ knows $v$'s ID; this is akin to forming TCP connections and establishing communication with a device whose IP address is known.  To facilitate edge formations, every new node $u$ that enters the network at a round $r$ is seeded with some limited knowledge of the network, typically just the ID of one other node $v$ present in rounds $r-1$ and $r$. The adversary must ensure that each $v$'s ID is given to only some $O(\log n)$ new nodes each round to avoid overloading $v$. Communication, edge formation, and limits on the amount of seed information will be formally discussed shortly.

 {\bf Overview of Problem Statement.} Our goal is to design a protocol that (whp) constructs and maintains an overlay network $\G = \{G_1, G_2, \ldots, G_r, \ldots \}$, where each $G_r=(V_r, E_r)$ denotes the overlay graph in round $r$. In particular, $E_r$ is the set of overlay edges created by the protocol. The protocol must
\begin{compactenum}
\item construct $\G$ during the bootstrap phase, and
\item maintain it in a well-defined robust manner despite churn for $\poly(n)$ rounds (whp).  
\end{compactenum}
Importantly, the overlay thus constructed and maintained must be sparse with degree bounded by $O(\log(n))$. We will prove several useful characteristics of the overlay $\G$ that we build. In particular, we show that protocols in the CONGEST model and node capacitated clique models designed for static networks  can be simulated in Spartan (see section~\ref{sec:simulation}).

{\bf Adversary's characteristics.} Our protocols operate against an almost adaptive adversary that subsumes the power of adversaries employed in both \cite{APRRU15} and \cite{DGS16}. We say that an adversary is {\em $t$-late omniscient} if, at each round $r$,  it is aware of all the algorithmic activities (including the random bits employed by the algorithm) from round 1 up to round $r-t$, but is oblivious to the execution and random bits employed by the algorithm from round $r-t+1$ onward. Such an adversary will, of course,  be aware of the overlay network $\G$ that the algorithm has constructed up till round $r-t$. Furthermore, we say that an adversary's {\em actions are instantaneous} if its decision to add or remove nodes take effect immediately without any forewarning. 
Our adversary's actions are instantaneous and it is assumed to be $\Theta(\log \log n)$-late omniscient. In contrast, \cite{APRRU15} employs the weaker oblivious adversary, but, like ours, its actions are instantaneous. The adversary in \cite{DGS16} is also $\Theta(\log \log n)$-late omniscient, but its actions are not instantaneous. Any changes to the network (both inserting and deleting nodes) must be announced to the protocol $\Theta(\log \log n)$-rounds before the change occurs.

{\bf Seed knowledge in new nodes.} If a new node is completely unaware of the rest of the network, it will be impossible for the node to integrate into the network. Therefore, we assume that new nodes are seeded with some very limited information. 

At the start of the {\em bootstrap phase}, we assume that each node in $V_1$ is seeded with $\Theta(\log n)$ IDs, each chosen uniformly at random (UAR) from IDs in $V_1$. This is a reasonable assumption. If Spartan is to be built from scratch, a centralized mechanism will be required -- and this assumption ensures that the role of the centralized mechanism can be limited to providing the random samples. On the other hand, we may wish to build Spartan from a pre-existing \Pe network. Random node IDs are typically easy to obtain in this scenario as well because most \Pe networks have good expansion and facilitate random walks that can mix fast -- a property that can be leveraged to obtain random sampling in a fully distributed manner. In fact, many of the existing works on \Pe networks including our own (see Section~\ref{sec:sampling-log}) provide efficient sampling procedures~\cite{GMS04,APR13,AMMPRU13,APRRU15}. 

Any node $u$ that enters the network at some round $r$ during the {\em maintenance phase} must be seeded with the ID of  at least one pre-existing node $v \in V_{r-1} \cap V_{r}$. Notice that this requires both $u$ and $v$ to be present during round $r$. Node $u$ can then contact $v$ and, in turn, $v$ can provide $u$ with some information that will help $u$ to become a well-connected node in the network.  However, if too many new nodes are seeded with the ID of the same node  $v$, then $v$ can be overwhelmed. So we limit the number of nodes seeded with the ID of any particular node to $O(1)$. 

 {\bf Overlay Edge Formation and Communication.} 
Each node $u$ is provisioned with $\Delta\in \Theta(\log(n))$ ports   through which overlay edges can be established and communication can take place. Suppose $u$ wants to establish an overlay edge with $v$. Then, $u$ must send an edge formation request  to $v$ and $v$ can either accept or ignore (thereby implicitly rejecting) the  request. Any overlay edge thus formed, in essence, connects a unique port in $u$ with a unique port in $v$. Consequently, such an overlay edge can be formed only if both $u$ and $v$ have ports to spare.    After the overlay edge $(u,v)$ is formed, either nodes can unilaterally drop the overlay edge at the end of any round, thereby freeing the associated ports in both nodes. If one of the endpoints $v$ of an overlay edge $(u,v)$ is churned out by the adversary, the edge is immediately dropped and $u$ is immediately aware of the port that has been freed up.

Nodes communicate and form overlay edges via message passing. A node $u$ can communicate with another node $v$  and/or form an overlay edge $(u,v)$ in $\G$ only if $u$ knows the ID of $v$; 
one can think of $v$'s IP address as playing the role of its ID. 
Node $u$, in addition to knowing the IDs of all its neighbors in $\G$, may also know IDs of nodes that it has received through messages. 
So $v$ is not required to be a neighbor in $\G$.  In order to avoid congestion, we limit the number of bits sent from $u$ to $v$ per round to be  at most $O(\plog(n))$ in size. Furthermore, no node can communicate with more than $O(\plog(n))$ other nodes per round. The protocol must ensure that no node sends more than $O(\plog(n))$ messages per round and no more than $O(\plog(n))$ other nodes send messages to any given node.\footnote{In fact, the only place where any node needs to send and receive $\omega(\log n)$ messages (but no more than $O(\plog(n))$ messages whp) is our random walks based sampling protocol (see Algorithm~\ref{alg:sampling}). This may however be unavoidable under $(\epsilon n, \Theta(\log \log n))$ churn. See the remark at the end of Section~\ref{sec:sampling-log} for more details.} 

The facility to communicate without forming edges may lead one to question why our model should even bother to form edges. We emphasize that there are several significant advantages to forming overlay edges. When $v$ is not an overlay  neighbor, $u$ may not know whether $v$ is still in the network. 
If $v$ is no longer in the network, $u$'s message to $v$ will simply be dropped. Moreover, $v$ will be guaranteed to receive messages sent by its overlay neighbors, but if $u$ isn't its overlay neighbor, then, $v$ may not receive the message sent by $u$; such situation may arise if the number of nodes sending messages to $v$ simultaneously exceeds the number of free ports in $v$. Finally -- and perhaps most importantly -- the very essence of building the overlay edges is to facilitate other protocols to run on the overlay network so constructed. Consequently, if we can prove that the overlay network $\G$ that we construct has certain properties, application designers can build their protocols upon such guarantees.

 {\bf Breakup of a round.} To avoid a fundamental impossibility (from Theorem 2 in~\cite{APRRU15}), each round permits a ``handshake" between two communicating nodes. 
So, if $u$ sends a message to $v$ in a round, then $v$ has the option of responding within the same round. 
Thus, in each round $r$, each node $u$ in sequence

\begin{compactitem}
\item[-] gets to know which ports are active (as  neighbors can drop edges or be churned out), 
\item[-]  performs local computation,
\item[-] sends up to $\plog(n)$ messages typically (but not limited to) requesting  for some information or requesting the formation of an edge,
\item[-] receives up to $\plog(n)$ messages sent by other nodes,
\item[-] performs local computation.
\item[-] sends up to $\plog(n)$ messages limited to responding to requests received, and finally
\item[-] receives up to $\plog(n)$ messages.
\end{compactitem}
Our protocol employs the second send/receive part of  rounds only when forming edges.  

%

\section{The Spartan Design and Applications}
\label{sec:spartan}
In this section, we first describe Spartan's design comprising $\Theta(n/\log n)$ committees (also called supernodes in literature) strung together in the form of a butterfly network. We then illustrate the power of this design to perform distributed computation on the network of committees. In particular, we show how algorithms in the CONGEST model of computation can be simulated in Spartan. We also describe how distributed hash tables can be implemented in Spartan.  In subsequent sections, we will  describe  how to build and maintain  Spartan in a concrete manner. 

\begin{figure}[htbp]
  \begin{center}
    \includegraphics[width=0.35\textwidth,clip=true,trim=150 50 250 45]{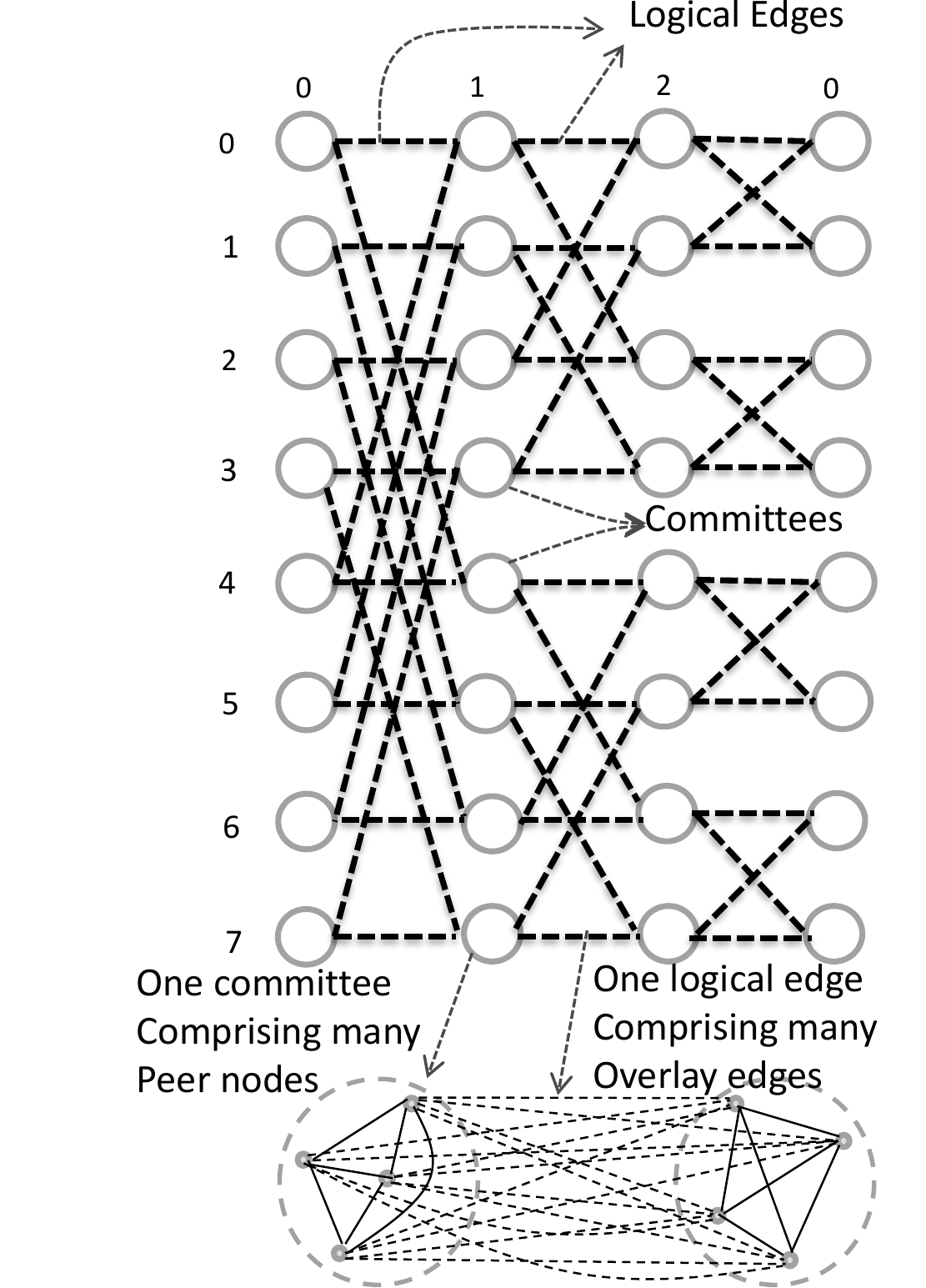}
  \end{center}
  \caption{A schematic depiction of a Spartan overlay built on a wrapped butterfly framework graph. Note that the first and last columns are the same set of committees since the butterfly is wrapped.}\label{fig:schematic}
\end{figure}

\subsection{Spartan Design} \label{sec:design}

The key building blocks of Spartan are committees. The Spartan architecture comprises a set $\C$ of  $\Theta(n/\log n)$ committees. Each committee consists of a dynamically changing set of member nodes chosen randomly with the guarantee that the committee has $\Theta(\log n)$ nodes as its members at any given time. The nodes within a committee are completely connected via overlay edges. Any pair of committees must be able to form {\em logical edges} between them. A logical edge between two committees is the complete bipartite set of edges between the nodes of the two committees. The set of committees $\C$ and the set of logical edges $\L$ together  form a wrapped butterfly framework graph $F=(\C, \L)$. The wrapped butterfly has the following structure. The committees in $\C$ (of cardinality $N  = k 2^k \in \Theta(n/\log n)$ for some $k$) are arranged in $2^k$ rows and $k$ columns. Thus, each committee $C \in \C$ can be addressed by a pair $(\r, \c)$, $0 \le \r < 2^k$ and $0 \le \c < k$,  denoting its row and column numbers, resp. Logical edges connect committees in a column $\c$ to committees in column $\c+1 \pmod k$. Two such  committees $(\r, \c)$ and $(\r', \c+1 \pmod k)$ are connected by a logical edge if and only if  either
\begin{compactenum}
\item both committees are in the same row, i.e., $\r = \r'$, or
\item $\r$ and $\r'$ differ at exactly the $(\c+1 \pmod k)$th bit in their binary representations.
\end{compactenum}

 Figure~\ref{fig:schematic} shows a schematic; for more details, please see Chapter 4 of~\cite{MU05}. The committees, logical edges, and $F$ must be built during the bootstrap phase and maintained (whp) for a sufficiently large $\poly(n)$ number of rounds. Since $F$ is a constant degree graph, the overall degree of each node (in terms of the number of incident overlay edges) is at most $O(\log n)$.
We require the Spartan protocols to create and maintain $\C$, $\L$, and therefore $F$.

The $\Theta(\log n)$ size of committees has two main advantages that we exploit. On the one hand,  the committees are large enough to ensure that an adversary that does not know the current members in the committee will (whp) be unable to disrupt it beyond a manageable level.  This advantage can however be lost quickly if we leave a committee unattended for $\Omega(\log \log n)$ rounds under a churn rate of $(\epsilon n,\log \log n)$ as all nodes can be replaced within $O(\log \log n)$ rounds.  Thus, nodes in the committees are dynamic. Each node moves to a new committee every $\Theta(\log \log n)$ rounds on expectation. In particular, unless churned out, a node stays in a committee for a number of rounds drawn from the geometric distribution with parameter $\Theta(1/\log \log n)$.

 On the other hand, the size of each committee is small enough that  the {\em set of  committee members at round $r$}  of a committee $C \in \C$ denoted $\ccg_r(C)$  can be encoded within $O(\log^2 n)$ bits. (We ignore the subscript when the round in question is clear from context.)

%
%

We say that a {\em committee $C$ has been  robust until round $r$} if (i) it was successfully created during the bootstrap phase, (ii) at every round $r'\in \{B, B+1, B+2, \ldots, r\}$, the number of members was at least $\Omega(\log n)$ nodes, and (iii) in every pair of consecutive rounds $r'-1$ and $r$, there is at least one common node in $C$ in both those rounds. Likewise, we say that a {\em logical edge between two committees $C_1$ and $C_2$ has been   robust until round $r$} if (i) both $C_1$ and $C_2$ have been robust until round $r$ and (ii) the logical edge was successfully created during the bootstrap phase and, at every round $r'\in \{B+1, B+2, \ldots, r\}$, there is an edge between every $u \in \ccg_{r'}(C_1)$ and $v \in \ccg_{r'}(C_2)$.  We say that a {\em Spartan implementation has been robust until round $r$} if all the $N = O(n/\log n)$ committees and the $O(N)$ logical edges have been robust until round $r$. The forthcoming sections are devoted to showing how we can construct and maintain Spartan for an arbitrarily large polynomial in $n$ rounds. In the rest of the section, we will briefly discuss how Spartan can be useful for distributed computation and distributed hash tables.

\subsection{Distributed Computation and Distributed Hash Tables in Spartan}\label{sec:simulation} While several works have addressed computation in highly dynamic networks, we take a significant step forward in designing a mechanism whereby standard algorithms  in established static distributed computing models  can be executed off-the-shelf. Additionally, distributed hash tables with $O(\log n)$ time storage and lookup times can be implemented despite heavy churn.

Given a Spartan network that is robust for a sufficiently long period of time, we can execute any CONGEST algorithm $\A$ that is designed to run on a butterfly network with $N = |\C| = O(n/\log n)$ nodes provided that the state of the nodes executing $\A$ remains small. Each committee $C$ takes up the role of executing $\A$ as if it was a node in a butterfly of $N$ nodes. However, the constitution of each committee can change with time, so care should be taken to ensure that any new node entering the committee should be updated with the current state of the node that is being simulated by the committee. This requires a minor restriction that the state of the nodes executing $\A$ must be bounded by $O(\plog(n))$ bits. (Of course, one can trade off this restriction with the churn rate.  CONGEST algorithms with larger states can be executed as long as the churn rate is slow enough to ensure that states can be copied into new nodes.) In addition, as per our model, committees can also directly communicate with each other even if they are not neighbors in the butterfly, provided that the nodes in each committee know the IDs of nodes in the other committee. This means that any algorithm designed for the node capacitated clique model~\cite{AGG+19} can also be executed.

Importantly, Spartan can be used to implement distributed hash tables. The committees serve as addressable locations within the butterfly network and any data item with a well-defined key value can be stored in the location given by the hashed value of its key. We give a brief overview here and refer to~\cite{JP13} for more details. Whenever a data item in the form of a $\langle \text{key, value}\rangle$ pair is to be stored in the network, the data item is routed to the location given by the hashed value of the key. Similarly, when the item  corresponding to a particular key is required, the request is routed to the location given by the hashed value of the key and retrieved from the committee in that location. Care should be taken to ensure that the data item is maintained properly in the committee. As in the case of simulating CONGEST algorithms, the data items stored in each committee should be copied to new nodes that enter the committee.

\section{Constructing Spartan During the Bootstrap Phase} \label{sec:warmup}
We will now present the steps to construct Spartan during the bootstrap phase.  Algorithm~\ref{alg:warmup-overview} provides the high level steps, each of which requires at most $O(\log n)$ rounds. Moreover, all these steps respect the congestion requirement in that no node will either send or receive more than $O(\log n)$ messages at any round; we will explicitly discuss this whenever it isn't obvious.  Each step is then explained and analyzed in detail subsequently.
\begin{algorithm}
\caption{Overview of approach to build a relaxed Spartan during the bootstrap phase.}\label{alg:warmup-overview}
\begin{algorithmic}[1]
\small
\STATE A leader node $\ell$ is elected taking advantage of a fraction of the $\Theta(\log n)$ seed random IDs that each node in $V_1$ possesses.  \label{lno:leader}
\STATE Construct an $O(\log n)$ height binary tree with vertex set $V_1$ and $\ell$ as root. (Recall $V_1=V_2=\cdots = V_B$.) \label{lno:binary} 
\STATE Compute the inorder traversal number for each node in $O(\log n)$ rounds (taking advantage of the height being $O(\log n)$) and form a cycle using the first $N = k 2^k \in O(n/\log n)$ nodes, each of which will eventually become committee leader. \label{lno:cycle}
\STATE This cycle is transformed in $O(\log n)$ rounds into a wrapped butterfly network with $k$ columns and $2^k$ rows.  \label{lno:butterfly}
\STATE The nodes that are not committee leaders will then randomly join one of the $N$ committees such that each committee has $\Theta(\log n)$ nodes. \label{lno:join}
\STATE The nodes within each committee form complete pairwise overlay edges.  \label{lno:formccg}
\STATE The committee leaders exchange $\ccg$s with their neighbors in the butterfly network using which the complete bipartite overlay edges can be formed between nodes in neighboring committees.  \label{lno:formedges}

\normalsize
\end{algorithmic}
\end{algorithm}

\noindent {\bf Line number  \ref{lno:leader} of Algorithm \ref{alg:warmup-overview}.}  We claim that a leader $\ell$ can be elected by the following elementary algorithm: each node generates a random number with a sufficiently large number of $\Theta(\log n)$ bits  and creates a message with the random number and its ID. It then choses $\Theta(\log n)$ nodes uniformly at random from the samples it has been seeded with (without replacement) and floods the message to those nodes. In every round thereafter, it continues to flood the message with the largest value seen so far. After a period of some $\Theta(\log n)$ rounds, we claim that (whp) only one message, that belonging to the node $\ell$ which generated the largest random number, will survive and $\ell$ would be elected leader.
\begin{lemma}
	The leader election algorithm employed in line number \ref{lno:leader} of Algorithm \ref{alg:warmup-overview} can, whp, correctly elect a unique leader $\ell$ in $O(\log n)$ rounds. Furthermore, each node consumes at most $O(\log n)$ of its seeded random ID samples.
\end{lemma}
\begin{proof}

	From the description of the protocol, it is clear that each node consumes at most $O(\log n)$ seeds from its samples. To see how the algorithm will take $O(\log n)$ rounds, consider this. Each node chooses $\Theta(\log n)$ nodes at random at the start of the protocol (and thereafter uses the same nodes) to flood its message. This means that we are looking at a random graph $G(n,p)$ with $p$ in $\frac{d\log n}{n}$, with $d>1$ being a sufficiently large constant. We know that for such a graph, the diameter is in $O(\log n)$, thus with high probability, $\ell$'s message will be flooded through the network in $(O(\log n))$ rounds.

\end{proof}

\noindent {\bf Line number  \ref{lno:binary} of Algorithm \ref{alg:warmup-overview}.}
The tree is created in two phases, both of which run for $O(\log n)$ rounds and each node consumes at most $O(\log n)$ of the seed tokens. The first phase begins with the tree being just the root node $\ell$. Each node in the tree begins with two vacant spots for its children. At each time step, each node in the tree that has at least $0<q \le 2$ vacancies queries $q$ nodes at random and requests them to become its children. Nodes not in the tree (called non-tree nodes), if requested, accept at most one such request and fill the vacancy and become nodes in the tree. The first phase lasts for a suitably large $\Theta(\log n)$ rounds after which (as shown in Lemma~\ref{lem:phase1}) the tree consists of at least $n/2$ nodes (whp).  Note that such a tree will have at least $n/2$ vacancies as well.

During the second phase, in each round, each non-tree node $u$ randomly probes a node $v$ (again by consuming one of its random samples) to see if it is a tree node with a vacancy.  If $v$ is a tree node and has a vacancy, then, it will respond positively to exactly one such request among the several it might have received. If $v$ responded positively to $u$, then $u$ becomes its child. Again, we will show in Lemma~\ref{lem:phase2} that within $O(\log n)$ rounds, (whp) the number of non-tree nodes dwindles to zero. 

We emphasize that in both phases, the sampling is from the global set of nodes. In the first phase (resp., second phase), samples that fall on the nodes already in the tree (resp., node not yet in the tree and tree nodes that already have two children) are ignored. Crucially, by a simple balls-into-bins argument, this ensures whp that no node receives $\omega(\log n)$  requests per round. Moreover, just the $\Theta(\log n)$ seed samples per node will be sufficient to execute line number  \ref{lno:binary} of Algorithm \ref{alg:warmup-overview}.

\begin{lemma} \label{lem:binary}
	For the procedure described in Line number \ref{lno:binary} in Algorithm \ref{alg:warmup-overview}, the following two statements hold:\\
	i) After a sufficiently large $\Theta(\log n)$ rounds of the first phase, (whp) the tree has grown to have at least $\lfloor n/2 \rfloor$ nodes in it.  \label{lem:phase1}
	\\
	ii) After a sufficiently large $\Theta(\log n)$ rounds of the second phase, (whp) the number of nodes not in the tree has dwindled down to zero.  \label{lem:phase2}
	\\
Because of i) and ii), after a a sufficiently large $\Theta(\log n)$ rounds, (whp) the tree has grown to have $n$ nodes in it. Furthermore, no node sends more than $O(1)$ messages and (whp) every node receives at most $O(\log n)$ messages.
\end{lemma}
\begin{proof}

We will first proceed to prove statement (i). Statement (ii) holds by symmetry. Let's fix a round $r$, let the number of nodes currently in the tree be $\tau$, this means that there are exactly $\tau+1$ vacancies and $n-\tau$ non tree nodes. We want to calculate exactly how many vacancies are being filled in a round. Let $X$ be the random variable that counts the number of vacancies that has been filled in a round $r$. Note that the event of an vacancy being filled by a node $u$ is precisely the same as node $u$ becoming a non-tree node. Which is also equivalent to node $u$ getting at least one invite. Thus, at any round $r$, counting the number of nodes that received an invite, gives us a measure on the number of nodes that have become a part of the tree in $r$ (i.e., $X$). Let $X_i$ be a random variable such that
 \begin{center}
 	\[ X_i = \begin{cases} 1 & \text{if non-tree node $i$ received at least 1 invite} \\ 0 & \text{otherwise.} \end{cases} \]
 \end{center}
Clearly, $X=\sum_{i=1}^{n-\tau} X_i$. At this juncture, it is important to note that the random variables $X_i$ are not mutually independent. However, as they are negatively associated (chapter 4 of Dubashi and Paconesi's book \cite{DP09}) (as we will briefly explain in the following paragraph) we are free to use various probabilistic bounds. 
To see how the random variables are negatively associated,  we may visualize the act of vacancies inviting nodes as a balls and bins model.  In the act of a vacancy inviting a node to be its child, the node with the vacancy ($u$) uses one of its random samples to invite another node (that's chosen UAR from the network) to fill its vacancy from among the $n$ nodes in the network. This is comparable to the act of throwing a single ball into $n$ bins. Clearly in this setting, a non tree node getting at least one invite is equivalent to a bin being non empty (i.e., getting at least one ball). We now introduce the following two useful lemmas:
\begin{lemma}
Consider the balls and bins model where $m$ balls have been thrown at $n$ bins ($m$ not necessarily equal to $n$). Let $B_i$ count the number of balls in bin $i$. The random variables $B_1,..., B_n$ are negatively associated. [$from$ Eg. 3.1 \cite{DP09}] \label{lem:bbm}
\end{lemma}
\begin{lemma} 
In the same balls and bins model introduced in Lemma~\ref{lem: tree}, under the Discrete Monotone Aggregation property, random variables that are non-decreasing or non-increasing functions of the negatively associated variables $B_1,...B_n$ are also negatively associated.  [$from$ Eg. 3.2 \cite{DP09}] \label{lem:dma}
\end{lemma}

From both lemma \ref{lem:bbm} and lemma \ref{lem:dma}, it's clear that the random variables $X_i$'s are negatively associated. Thus, for any value of $k$, when $2k$ balls are thrown at $n$ bins, 
\begin{equation}
    \mathbf{Pr}(\text{less than $k$ bins are filled}) \leq {n \choose k} \frac{k^{2k}}{n^{2k}}.
\end{equation} 
When there are $\tau$ nodes in the tree, and therefore $\tau+1$ vacancies:

\begin{center}
$\textbf{Pr}( X \leq (\tau+1)/2  )\leq {n\choose {(\tau+1)/2}}$ $({\frac{(\tau+1)/2}{n}})^{(\tau+1)}$ \label{eqn:one}

\end{center}
Now we may the use the following standard inequality that states:
\begin{center}
	\begin{align*}
	\forall k,   1\leq k \leq n, \frac{n^k}{k^k} \leq {n \choose k} \leq  \frac{e^kn^k}{k^k}
	\end{align*}

\end{center}
To get: 
\begin{center}
	\begin{align*}
	\mathbf{Pr}(X \leq \frac{\tau+1}{2}) & \leq \left(\frac{e^{\frac{\tau+1}{2}}n^{\frac{\tau+1}{2}}}{\frac{\tau+1}{2}^{\frac{\tau+1}{2}}}\right) \left(\frac{\frac{\tau+1}{2}}{n}\right)^{\tau+1}
	\\
	\\ 
	& \leq \frac{e^{(\tau+1)/2}}{4^{(\tau+1)/2}}{\left(\frac{2(\tau+1)}{n}\right)}^{(\tau+1)/2}
		\end{align*}
\end{center}


 

For any $d>0$ and $\tau\leq n^{1/d}$, since the fraction $\frac{2(\tau+1)}{n}$ approaches $\frac{1}{n^{1-1/d}}$, this means that more than half the vacancies are filled with high probability. But, as $\tau$ approaches a fraction of $n$ (say $n/\lambda$ for any $\lambda$, $0<\lambda\leq 2$),  the above probability approaches $1-(1/\lambda)^{\tau}$ which does not guarantee high probability. However,  even as $\tau$ approaches a fraction of $n$,  at least one fourth  of the vacancies are filled in every round with high probability.

  Recall that $X= \sum_{i}^{n-\tau} X_i$. We have $\mathbf{Pr}(X_i=1)$ is $ 1-(1-1/n)^{\tau+1}$. The expected number of vacancies filled in a round is then $\mathbf{E}[X]= (n-\tau) \left( 1-(1-1/n)^{\tau+1}\right) $. We first use the binomial expansion to obtain the following upper bound:
  \begin{center}
  	\begin{align*}
  (n-\tau)(1-1/n)^{\tau+1} \leq (n-\tau) \left[1-\frac{\tau+1}{n}+\frac{(\tau+1)(\tau)}{2!n^2}\right]
  	\end{align*}
  \end{center}
    Then we use the above bound to obtain a lower bound on the expectation as follows:
%
\begin{center}
	\begin{align*}
	\mathbf{E}[X] \geq \frac{3(\tau+1)}{8}
	\end{align*}

\end{center}
We will now use the following  useful variation of the chernoff bound from Mitzenmacher and Upfal~\cite{MU05} (where $\mu$ denotes the expectation of random variable $X$)
\begin{center}
	\begin{align*}
	\mathbf{Pr}(X\leq(1-\delta)\mu) < e^{-\frac{\mu\delta^2}{2}}, \text{for any }  0<\delta \leq 1  \\
	\end{align*}
	
\end{center}

 Since $\delta\mu \geq \frac{\tau+1}{8}$ and thus $\mu^2\delta^2\geq\frac{{(\tau+1)}^2}{64}$, the probability that at most one fourth of the vacancies are filled in a round $r$ is given as:
%
%
%
\begin{center}
	\begin{align*}
		\mathbf{Pr}(X \leq \frac{(\tau+1)}{4}) \leq \mathbf{exp}{\left\lbrace- \frac{(\tau+1)}{128}\right\rbrace}
	\end{align*}
\end{center}
Thus, for each round $r$, where $1<\tau+1 < n/2$, whp, at least a fourth of the vacancies are filled. Call such a round in which one fourth of the vacancies are filled a good round. Since for each round $r$, the probability of it being a good round is high, for given values of $n$, we will be able to calculate the value $R$ such that when we run the algorithm for $R\log n$ rounds, there are at least $10 \log n$ good rounds with high probability. Thus in $R\log n$ rounds, the number of nodes in the tree has grown to $n/2$. Statement ii) holds by symmetry, as we are looking at the same balls and bins scenario as above but with the tree nodes and non-tree nodes switched (as in the second phase, the non-tree nodes are looking for vacancies). Therefore, the same argument holds (only in reverse as we start with at most $n/2$ non-tree nodes that slowly dwindle down to zero). Thus in $\Theta(\log n)$ rounds, we have a tree on $n$ nodes with high probability.

Finally, it is clear from the protocol that each node only sends $O(1)$ messages, but it is not immediately clear whether they only receive $O(\log n)$ messages (whp). To show this upper limit on the number of messages received, we note that -- in each round of both phases -- nodes send messages to randomly chosen other nodes. So this process can again be viewed as a balls-into-bins process with at most $n/2$ balls being thrown randomly into $n$ bins. So no bin will receive any more than $O(\log n)$ balls (whp). This balls-into-bins limit implies that no node will receive more than $O(\log n)$ messages at any round. 
\end{proof}

\begin{lemma} \label{lem: tree}
There exists an $O(\log n)$ round implementation of Line number \ref{lno:binary} of Algorithm \ref{alg:warmup-overview} that (whp) results in a binary tree whose height is $O(\log n)$. 
\end{lemma}
\begin{proof}
From lemma \ref{lem:binary} it follows that there exists a procedure that can in $O(\log n)$ (whp) build a binary tree that includes all the nodes in the network. The height constraint follows from the fact that in any round the height of the tree being built can be increased by at most 1. Since the the procedure is terminated after $O(\log n)$ rounds, it follows that the height of the tree can be at most $O(\log n)$.
\end{proof}

{\bf Line number  \ref{lno:cycle} of Algorithm \ref{alg:warmup-overview}.}  The inorder traversal number is quite easy to compute. First we use a bottom-up convergecast, in which each node (starting with the leaf nodes)  sends up the number of children rooted at their sub-trees.  This ensures that each node in the tree knows exactly how many descendants are in each of their left and right subtrees after a period of $\Theta(\log n)$ rounds. 

When the root gets this information, it can compute its inorder traversal number as $1+L$, where $L$ is the number of nodes in it's left subtree. Notice that all the nodes in the left subtree will have inorder traversal numbers in the range $[1, L]$ while the nodes in the right subtree will have inorder traversal numbers in the range $[L+2, n]$. The root then passes these ranges to its appropriate children. We can continue this process in a top-down recursion such that each node gets to know its inorder traversal number and passes on the appropriate range of inorder traversal numbers to its two children.
 Since the number of committees $N \in \Theta(n/\log n)$ can be precomputed, the first $N$ nodes (in the inorder traversal ordering) can identify themselves as committee leaders. 

Now, we again perform a bottom up process in which the $N$ committee leaders organize themselves as a ring. Observe that a ring (along the lines of a circular linked list) can be identified by the address of one arbitrary node in the ring called the head; each node within the ring has an overlay edge pointing to its clockwise neighbor and another to its counter-clockwise neighbor. First, the leaves that are also committee leaders form a trivial ring with just one node and pass the address of the head (here the head is themselves) to their parent. The parent then forms a new cycle by merging the (up to) two cycles (through the use of their cycle heads) it received from its two children. Note that it includes itself as a part of the cycle only if it is also a committee leader otherwise it just passes the head of the merged cycle up to its own parent. All nodes also ensure that the head of the cycle is always the node with the smallest inorder traversal number. This process continues until the protocol reaches the root, where the root puts together the two cycles from its left and right subtree (it includes itself in the $1+L$th position if it is a committee leader) and creates a contiguous cycle on $N$ nodes with the labels decided by the inorder traversal. The rest of the nodes (including the root if it is not part of $N$ nodes may now be discarded). Clearly, in $O(D)$ rounds, where $D$ is the depth of the tree, the nodes would have all organized themselves into a cycle. We know from lemma \ref{lem: tree} that the height of the tree is $O(\log n)$  Thus, 

\begin{lemma}\label{lem:committes}
There exists an $O(\log n)$ round implementation of line number  \ref{lno:cycle} of Algorithm \ref{alg:warmup-overview} that selects a set of $N \in \Theta(n/\log n)$ nodes and arranges them  as a  ring. 
\end{lemma}

{\bf Line number  \ref{lno:butterfly} of Algorithm \ref{alg:warmup-overview}.}  Before we start constructing the butterfly network, we must preprocess the ring of committee leaders in three steps to reach a suitable state depicted in  Figure \ref{fig:grid}.  

First, the ring is arranged in the form of a 2D grid comprising $k$ columns and $2^k$ rows as shown in Figure \ref{fig:grid} (i). This is easy because each  node with inorder traversal number, say, $i$ can position itself in row $\lfloor (i-1)/k \rfloor$ and column $(i-1) \mod k$. It has to then connect with nodes with inorder traversal numbers $i-k$ (when $i-1 \not\equiv 0 \pmod k$), $i$ (when $i \not \equiv 0 \pmod k$), $i-k$ (when $i>k$), and $i+k$ (when $i < N-k$). Note that connecting with $i-k$ and $i+k$ will require $\omega(1)$ rounds, but no more than $k \in O(\log n)$ rounds. This will allow us to identify each node in the grid by the pair $(\r, \c)$, where $0 \le \r \le 2^k - 1$ is the row number and $0 \le \c \le k-1$ is the column number.

\begin{figure}[htbp]
\begin{center}
\includegraphics[clip=true,trim=30 118 30 170,scale=0.5]{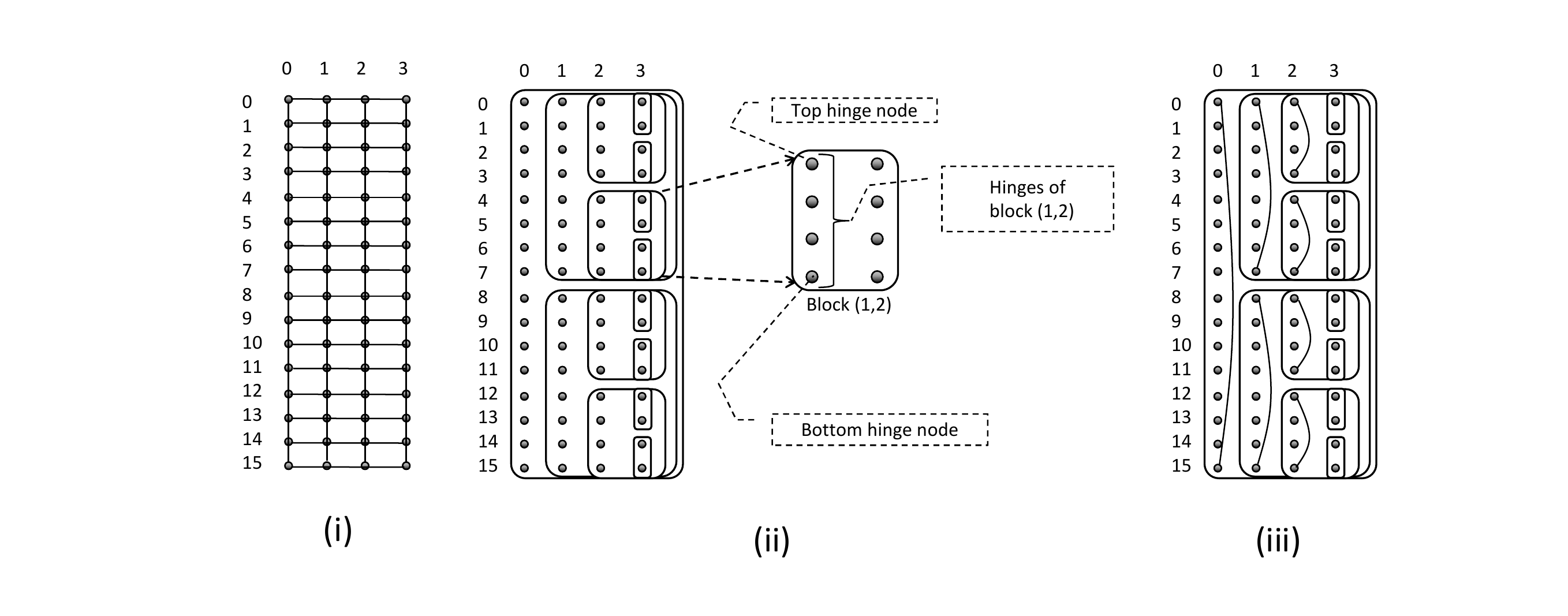}
\caption{preprocessing the cycle of committee leaders into a suitable grid structure with nested blocks that will serve as a starting point for the butterfly formation algorithm.}
\label{fig:grid}
\end{center}
\end{figure}

In the second step of the preprossessing stage,  we first structure the grid into nested blocks as defined shortly; see  Figure \ref{fig:grid} (ii). Blocks are suitably nested in rectangular subgrids of some $\kappa$ columns and $2^\kappa$ rows, $1\le \kappa \le k$. The nodes in the leftmost column of each block are called its {\em hinge nodes} or just {\em hinges}. If the hinges of a block occur in some column $\c$ in the grid, we say that the block is {\em hinged} at column $\c$. Column $\c$ in the grid  has $2^\c$ blocks hinged in it. Therefore, we refer to the $\rho$th block (from the top, starting at $\rho = 0$) hinged in column $\c$ as {\em block $(\rho, \c)$}, $0 \le \rho \le 2^\c -1$.  Let $\kappa_\c \in \{1, 2, \ldots k\}$ such that $\kappa_\c \equiv -\c \pmod k$; then, every block hinged in column $\c$ has $2^{\kappa_\c}$ rows. Block $(\rho, \c)$ thus refers to the rectangular subgrid with node $(\rho 2^{\kappa_\c}, \c)$ as the top-left node and $((\rho+1) 2^{\kappa_\c}, k-1)$ as the bottom-right node.

We will now describe the second step of the preprocessing stage in which the topmost hinge vertex must be connected to the bottom most hinge vertex in every block; see  Figure \ref{fig:grid}. Clearly, this is already true in blocks hinged at column $k-1$; this sets the base case for an inductive procedure. Consider a block $P$ (for parent) hinged at column $\c$ that has two smaller blocks $U$ (for upper) and $L$ (for lower) nested within it, both hinged at column $\c+1$. Once the top and bottom hinge nodes of $U$ and $L$ are connected via overlay edges, within a constant number of rounds ($5$ to be exact), the top and bottom hinge nodes of block $P$ can be connected in the following manner. The top hinge node of $P$ passes the request to the top hinge node in $U$ (as they are adjacent in the grid), which in turn passes the request to $U$'s bottom hinge. From there, the request is sent to the top hinge node of $L$ (which is right below $U$'s bottom hinge in the grid), which can then pass it on to $L$'s bottom hinge. From there, the intended final recipient is just one hop to the left in the grid. This means that after $5k$ (and thus $O(k)$) rounds,  ensure that the top hinge node of every block is connected to its corresponding bottom hinge node. 

The third (and final) step in the preprocessing stage is quite a small step. For every pair of blocks $U$ and $L$ hinged at the same column and nested within the same parent $P$, we ensure that their respective top hinge nodes are connected and likewise that their bottom hinge nodes are also connected. This is an easy $O(1)$ rounds step since the top and bottom hinges are connected for every block; see Figure~\ref{fig:butterfly} (i) \& (ii).

{\bf Butterfly construction.}\label{dis:butterfly} To construct the butterfly network, we perform the following steps in parallel at every parent block $P$ that has two blocks $U$ and $L$ nested within it. Consider one such example as shown in Figure~\ref{fig:butterfly} (ii) and, equivalently, in (iii), wherein, two blocks $U$ and $L$ consisting of $\kappa$ columns and $2^\kappa$ rows are initially connected at their respective top hinges and their respective bottom hinges. Number the hinges of both blocks $L$ and $U$ top to bottom from 0 to $2^\kappa-1$.  The final goal is to reach the configuration shown in Figure~\ref{fig:butterfly} (x), but an important intermediate goal is to get the $r$th pair of hinges from $U$ and $L$ connected for $0 \le r \le 2^\kappa - 1$ (shown in Figure~\ref{fig:butterfly} (viii)). One can easily achieve this in $\Theta(2^\kappa)$ rounds, but this will translate to an overall $\Theta(2^k)$ rounds when we consider the pair of blocks nested within the outermost block. But, with care, this can be achieved in $O(\kappa)$ rounds (and thus we will need $O(k)$ rounds in total). 

\begin{figure}[htbp]
\begin{center}
\includegraphics[clip=true,trim=0 0 0 0,scale=0.5]{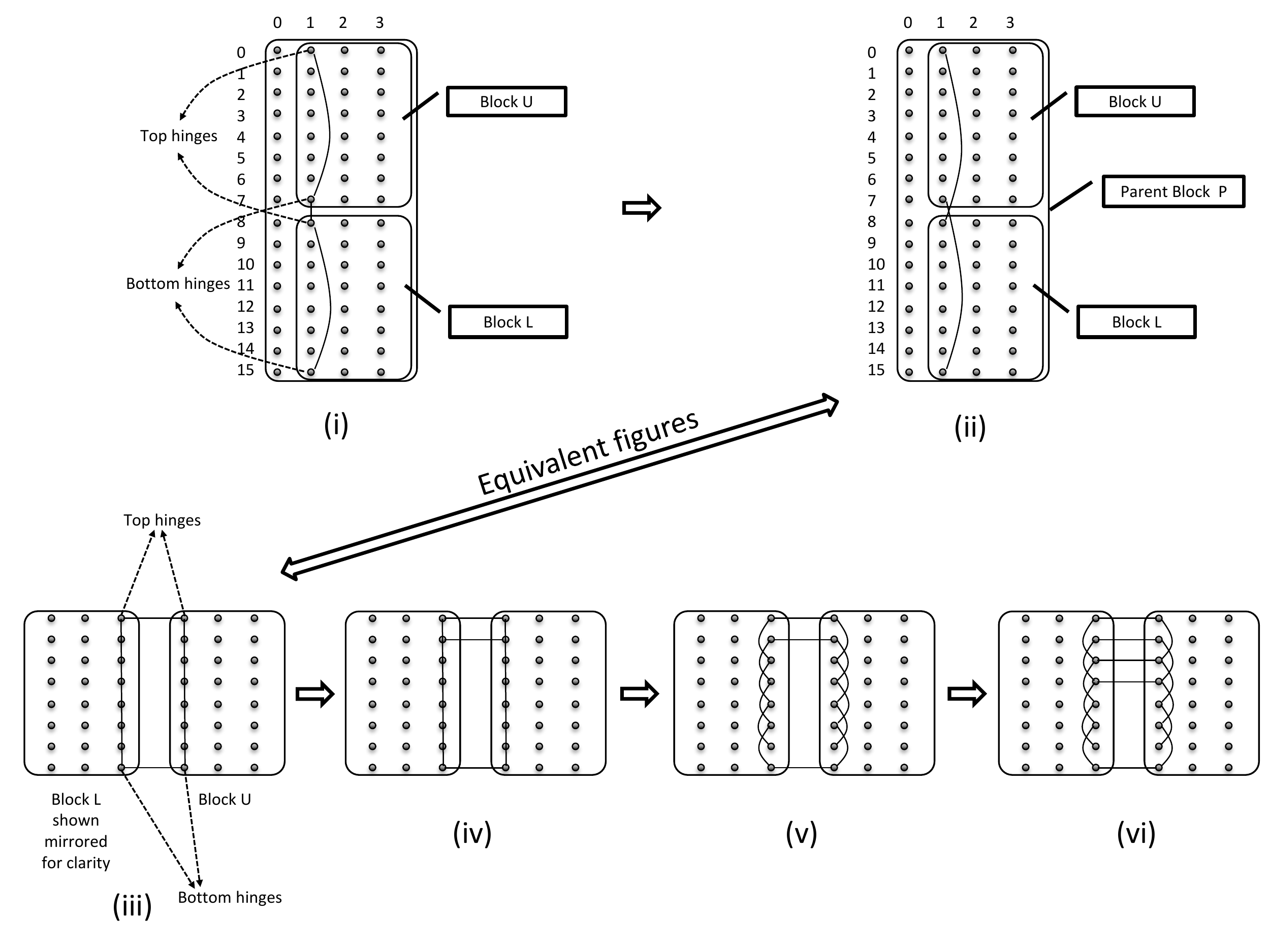}
\includegraphics[clip=true,trim=0 0 0 50,scale=0.5]{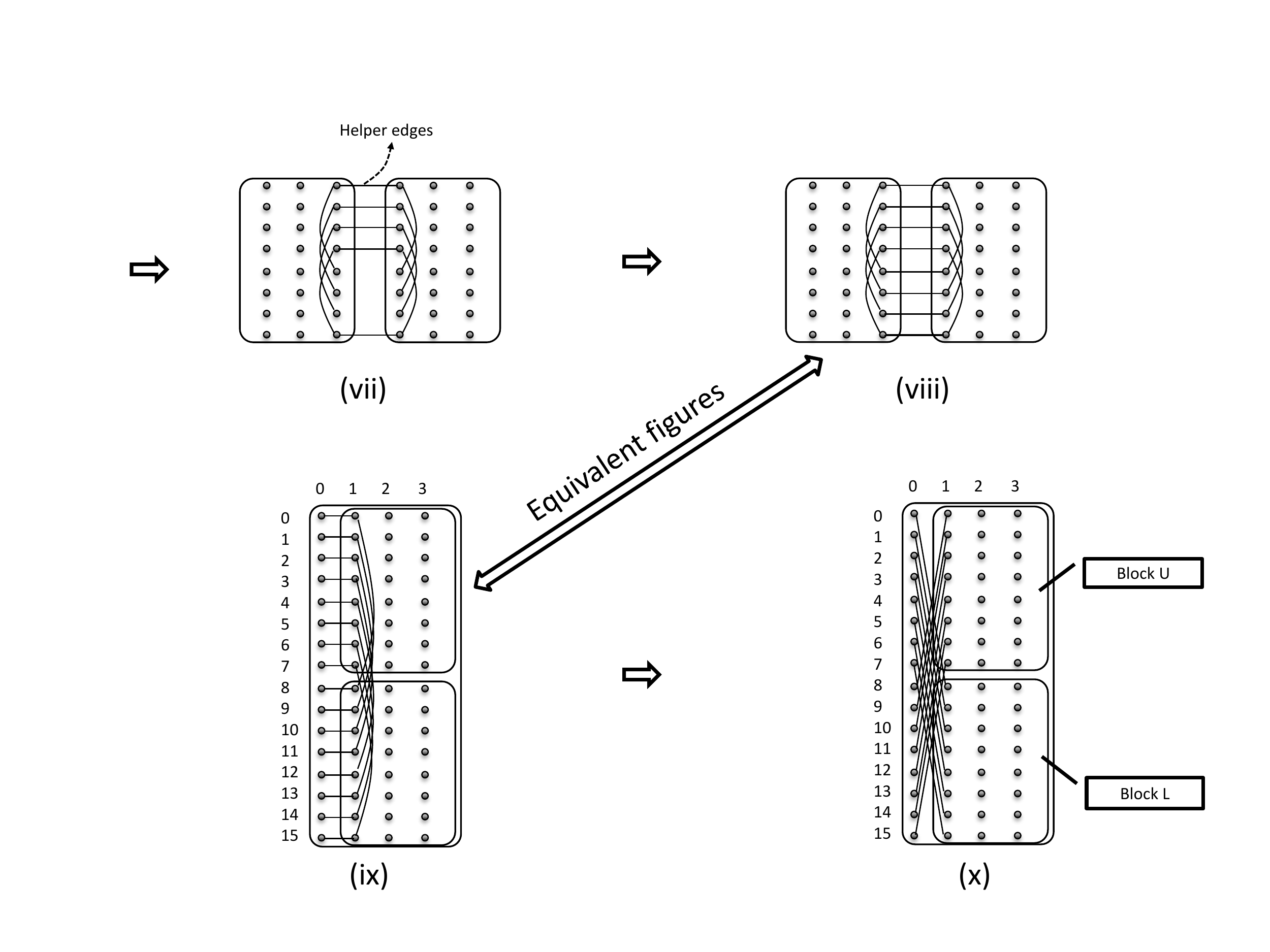}
\caption{The third (and final) step in the preprocessing stage is shown as the transformation  from (i) to (ii). The rest of the figures show the actual construction of the butterfly network. To better illustrate the progression of the butterfly construction, block $L$ is redrawn in a mirrored fashion from (iii) through (viii). In particular, (ii) is redrawn in (iii) with $L$ mirrored. Similarly, (viii) is redrawn in (ix) with $L$ drawn in the usual manner in which butterfly networks are drawn. To reduce clutter, we have not shown all the grid edges, but rather just show those edges needed to perform the next step.}
\label{fig:butterfly}
\end{center}
\end{figure}

Our approach involves $O(\kappa)$ stages (in total) with each stage requiring $O(1)$ rounds.   After each stage $i$, we ensure that the first $2^i$ pairs, i.e., $r$th hinges in $U$ and $L$ for $0 \le r < 2^i$, are connected (as shown in Figure~\ref{fig:butterfly} (iv) to (viii)); in particular, note (iii), (iv), (vi), and (viii). In addition, to help the next stage, we also ensure that in each block the $r$th hinge is connected to the $(r+2^i)$th hinge node; these are called {\em helper edges}.  In stage $i+1$, we wish to get the $r$th hinge nodes from $U$ and $L$ connected for $2^i+1 \le r \le 2^{i+1}$. This can be achieved in three steps because the $r$th node in $U$ can (via a helper edge) send its address to $(r - 2^i)$ hinge in $U$. Since $(r - 2^i)$th pair of hinge nodes are already connected in a previous stage, the address can be passed on to the $(r - 2^i)$ hinge in $L$. Again, via a helper edge, the address can be transmitted to the $r$th hinge in $L$, thereby enabling the $r$th hinge in $L$ to establish an edge with the $r$th hinge in $U$. We are not done yet because the helper edges in the next stage must have twice the reach as the helper edges in the current stage. For this purpose, we consider the two helper edges incident at some $r$th hinge in (say) $U$ leading to $r' = (r - 2^i)$th  hinge and $r'' = (r + 2^i)$th  hinge. Clearly, the $r'$th hinge can send its address to the $r''$th hinge in two steps via the two helpers we considered. Thus, $r''$th hinge can establish an edge with the $r'$th hinge as required. Of course, this will have to be repeated in parallel for all $2^i \le r \le \kappa - 2^i$ and for $L$ as well. Once the helpers for the next stage have been constructed, all other previously constructed helpers (barring those that are also grid edges) can be deleted. This will bring us to (viii) (or equivalently (ix)) in Figure~\ref{fig:butterfly}. 

Finally, to get the butterfly edges, notice that the first $2^\kappa$ hinges in $P$ (from top to the half way mark) must connect with the $2^\kappa$ hinges in $L$ and the second $2^\kappa$ hinges in $P$ must connect to the $2^\kappa$ hinges in $U$. More precisely, for $0 \le r \le 2^\kappa-1$, the $r$th hinge in $P$ must connect to the $r$th hinge in $L$. For $2^\kappa \le r \le 2^{\kappa+1} - 1$, the $r$th hinge in $P$ must connect with the $(r - 2^\kappa)$th hinge in $U$. This can be easily achieved in $O(1)$ rounds because the endpoints that must be connected are within two hops of each other. For example, for some $r$ such that $2^\kappa \le r \le 2^{\kappa+1} - 1$, there is a grid edge to the $r$th hinge in $U$ and since the $r$th hinges in $U$ and $L$ are connected, the $r$th hinge in $P$ can transmit its address to the $r$th hinge in $L$ in 2 rounds, which will then allow the $r$th hinge in $L$ to establish an edge with the $r$th hinge in $P$ in one round. 

Finally, we note from the description of the protocol that the congestion limits are not violated. Thus,
\begin{lemma}\label{lem:butterfly}
There exists an $O(\log n)$ round implementation of line number  \ref{lno:butterfly} of Algorithm \ref{alg:warmup-overview} that does not violated congestion limits and takes a ring of  $N=k2^k \in \Theta(n/\log n)$ nodes and arranges them  as a  wrapped butterfly network of $k$ columns and $2^k$ rows. 
\end{lemma}

{\bf Line number  \ref{lno:join} of Algorithm \ref{alg:warmup-overview}.}\label{dis:join}  We will reuse ideas from line number~\ref{lno:binary} to implement line number~\ref{lno:join} and, just as before, we will implement the line in two stages. The initial stage ensures that each committee gets $\Theta(\log n)$ nodes with high probability. This is done by each committee leader actively trying to recruit one member for every round for a period of $O(\log n)$ rounds. Once the first stage is finished, the second stage ensures that any unattached node can find a committee to join within $O(\log n)$ rounds. 
Let $\bgt = 3n/4N$. In the first stage, each committee will garner exactly $\bgt$ nodes in the following manner. Each committee leader uses its samples to invite one random node per round until it inducts $\bgt$ nodes into its committee. Consider a committee $C$ with leader node $c$. At each round, there are at least $n/4 - N \in \Omega(n)$ nodes that did not receive any invitation from  any other committee leader. This means that the probability $p$ with which committee leader may find a node for its committee is at least $\left(\frac{n/4-N}{n}\right)(1-1/n)^N$.  That is, when $N$ is $\frac{n}{\zeta\log n},\ \zeta \geq 1$, then $p$ is at least ($3/16$).
%
%

If we define $X_i$ to be the random variable that counts the number of rounds before a committee  leader can get its $i^{th}$ node after it got its $i-1^{th}$ one, then clearly $\sum_{i=1}^{\bgt} X_i$ gives the number of rounds required for a committee leader to get $\bgt$ nodes. Now these $X_is$ are not independent. Let's look at the construction of a similar random process, such as a coin toss in which you count the number of rounds before you get $\bgt$ heads. In this construction, a coin turns up heads with probability exactly $3/16$. Now this is a more pessimistic version of the same events described above and hence will take longer to reach the required goal of $\bgt$ heads. In the described experiment, this means that on expectation it takes at most $16/3$ rounds for the coin to turn up its $i^{th}$ head after its $i-1^{th}$ one, $1\leq i \leq \bgt$, which then means that on expectation it takes at most $\frac{16}{3}\bgt \in \Theta(\log n)$ rounds to get $\bgt$ number of heads.  We may then use Chernoff bounds to show that for any value of $n$ and any fixed value of $R,\ R \geq 6\mathbf{E}[X]$, the probability it will take more than $R\log n$ rounds to reach $\bgt$ number of heads is at most $\frac{1}{n^{R}}$ (from Theorem 4.4 \cite{MU05}). Going back to our original scenario, this means that for any committee leader it will take at most $O(\log n)$ rounds with probability $\frac{1}{n^{R}}$. We may now use the union bound to show that all committee leaders can get the required number of nodes in $\Theta(\log n)$ rounds, with probability at least $1-1/n^{R-1}$ for any fixed $R$.


After $\Theta(\log n)$ rounds, there can be still at most $n/4$ nodes left in the network that are not part of the committee. At this point, we begin the second stage, where each unattached node $u$ tries to join a committee. Each node that has become part of the committee has a budget of 1 node which it can use to induct into its committee. In the second stage, in every round an unattached node $u$ tries to probe a random node $v$ to become a part of $v's$ committee. If $v$ has not exhausted its budget then it accepts $u's$ request, otherwise $u$ tries again with a different node. Since even if all unattached nodes become attached there will still be $n/2$ nodes that would not have exhausted their budget, each un-attached node can find a committee with probability at least $1/2$. In a manner similar to above, this means that an unattached node $u$ can clearly succeed in finding a committee, whp, in $\Theta(\log n)$ rounds. Again, using union bound, we can guarantee that, whp, each unattached node can find and become a part of a committee $\Theta(\log n)$ rounds. Note that the size of any committee can at most double at the end of the second stage as described here (since each committee member can accept exactly one non-committee member), so no committee will end with more than $\Theta(\log n)$ members. Thus,
\begin{lemma}\label{lem:comsize}
There exists an implementation of line number  \ref{lno:join} of Algorithm \ref{alg:warmup-overview} that (whp) requires $O(\log n)$ rounds and ensures that every committee  gets  $\Theta(\log n)$ member nodes. 
\end{lemma}

{\bf Line number  \ref{lno:formccg} of Algorithm \ref{alg:warmup-overview}.}  Since all nodes know their respective committee leader's ID, they will send their own ID to their leaders. The committee leader will form a list of all nodes in the committee and communicate this to all nodes in its own committee. The individual nodes will then establish overlay edges with all other members in the committee.  At this point, the committee leader can construct the $\ccg$ and disseminate to all nodes in the committee

{\bf Line number  \ref{lno:formedges} of Algorithm \ref{alg:warmup-overview}.}  
The committee leaders exchange $\ccg$s with their neighbors in the butterfly network. Each committee leader in turn passes on the $\ccg$s it received to nodes in its own committee.  If $C_1$ and $C_2$ are two neighboring committees in the butterfly network, then the logical edge between $C_1$ and $C_2$ comprises overlay edges between every pair in $C_1 \times C_2$; here we are abusing notation and using $C_1$ and $C_2$ to refer to the committees as well as the set of nodes that make up each resp. committee.


\section{Maintaining Spartan} \label{sec:maintain}
We now turn our attention to the maintenance phase when churn can significantly impact the network. Maintaining Spartan comprises two parts:  the primary maintenance protocol and the supporting random walks based sampling protocol. These two interdependent  parts execute in tandem.  

Recall that each committee $C$ is identified by its row and column numbers in the Spartan structure. Thus we can always pick a random committee just by choosing a random row and column in the Spartan structure. However that will not be useful for communicating to that random committee because for some node $u$ to send a message to members of $C$, $u$ will require $C$'s current committee members list \cml(C). The random walks based sampling procedure that we provide ensures that each node  gets $\Theta(\log \log n)$  \cml s of randomly chosen committees every $\Theta(\log \log n)$ rounds. Using these random \cml~samples, the maintenance protocol ensures that Spartan is robust (whp) for an arbitrarily large $\poly(n)$ number of rounds.

\subsection{Random Walks Based Sampling} \label{sec:sampling-log}
We now show how each committee in Spartan can obtain the \cml s of $\Theta(\alpha \log n)$ uniformly and independently chosen committees using random walks for any $\alpha \in O(\plog(n))$;  in Section~\ref{subsec:maintain}, we set $\alpha = \log \log n$ in order to obtain $\Theta(\log n \log \log n)$ \cml~samples.   

Before we describe the procedure, we clarify a couple of key issues. Firstly, this sampling procedure assumes that Spartan is robust. Secondly, the samples produced by this procedure comes with a shelf life of $O(\log \log n)$ rounds because eventually, the nodes in the committees either move away to other committees or just churned out. To formalize this, we say that a sample $\cml(C)$ is valid  at a given round $r$ if the intersection between the list of IDs in $\cml(C)$ and the IDs of nodes in $C$ at round $r$ is of cardinality at least $\Omega(\log n)$. We later show that (whp) Spartan committees retain at least $\Omega(\log n)$ of their members for $\Omega(\log \log n)$ rounds (see Lemma~\ref{lem:robust}).

Normally, the random walks would require $O(\log n)$ steps to reach a random node, but we adapt the well-known pointer doubling technique~\cite{jaja,DGS16} to our context to achieve an exponential speedup of $O(\log \log n)$ steps. The steps are described in Algorithm~\ref{alg:sampling} under the assumption that Spartan is robust during its execution, which is reasonable because the sampling procedure is designed to work in tandem with the maintenance protocol that will guarantee robustness.

\begin{algorithm}
	\caption{Random walks based Sampling. 
	}\label{alg:sampling}
	\begin{algorithmic}[1]
		\small
		
		\item[{\bf Note: }] For clarity, the protocol is described from the perspective\\ of a single arbitrary committee $C$, but the protocol must be simultaneously executed by all committees. For simplicity, the round number starts at 1.
		
		\REQUIRE The committee $C=(\r,\c)$ generates $L=\Theta(\alpha \log^3 n)$ random walk tokens for some $\alpha \in \Omega(1)$. Each token $t$ has three data fields: a source committee field $t.src$, a destination committee field $t.dst$, and a random number field $t.rand$. For each $t$ generated by $C$, $t.src$ contains  $\cml(C)$, i.e., its source committee.  The random number field $t.rand$ is drawn from the geometric distribution with parameter $1/2$. The fields $t.src$ and $t.rand$ are immutable. The destination committee changes over time and its $\cml$  is stored in $t.dst$.
		
	\item[{\bf Invariant:}] Each token $t$ must be stored in both its source and destination committees. More precisely, each member node currently in the lists $t.src$ and $t.dst$ must contain a copy of $t$. Moreover, the lists $t.src$ and $t.dst$ must be as updated as possible. \COMMENT{One (handshake) communication round will be required to get the latest $\cml$'s from $C'$ and $C''$ and to ensure that $t$ is stored in both $t.src$ and $t.dst$.} \label{lno:invariant}

		\ENSURE In every row $\r$, the committees in row $\r$ are assumed (in  line~\ref{lno:col}) to know each others' $\cml$'s. This preprocessing step can be achieved by a combination of flooding and path doubling and  can be achieved in   $\log k \in  O(\log \log n)$ rounds. 
		 
		\item[\noindent {\bf At Round $r=1$:}]

		\STATE Let $C'$ and $C''$ be the two neighbours of $C$ at column $(\c+1) \pmod k$. Then, $t.dst$ is initialized to either $\cml(C')$ or $\cml(C'')$ chosen uniformly and independently at random. Ensure that the invariant is maintained. 
		
		\FOR{{\bf round numbers} $r=2$ to $r=1+\log k  \in O(\log \log n)$}
			\STATE Let $F$ be the set of tokens such that for each $f \in F$, $f.dst = \cml(C)$ and $f.rand \ge r$. Let $S$ be the set of tokens such that for each $s \in S$, $s.src = \cml(C)$ and $s.rand = r-1$. 
			\STATE Match each $f \in F$ with a unique $s \in S$ chosen uniformly at random (as long as such a unique $s$ is available). If $|F| \ne |S|$, there will be some unmatched tokens in either $F$ or $S$ (depending on whether $|F| > |S|$ or $|S| > |F|$) that are unmatched. Discard all such unmatched tokens in $F$ or $S$. Thus, for rest of the current iteration, $|F| = |S|$.
			\STATE For each matched pair $f$ and $s$, reassign $f.dst \leftarrow s.dst$. Ensure  invariant in line number~\ref{lno:invariant} for $f$. \COMMENT{This will require one (handshake) communication round to ensure that the updated $f$ is sent to $f.src$ and $f.dst$.}
			\STATE Discard tokens in $S$.
			\STATE $L \leftarrow L/4$ .
		\ENDFOR{\{At this point, the destination committee of each surviving token $t$ (i.e., $t$ has not been discarded so far) has a row number that is chosen uniformly at random from $[1, 2^k]$, but its column number is not random. We will fix this in the following loop.\}}
		\FOR{each token $t$ with $t.dst = \cml(C)$}
			\STATE $t.dst \leftarrow \cml(C')$ where $C'$ is chosen uniformly at random from the set of committees in row $\r$ (where $C$ belongs). \label{lno:col} 
		\ENDFOR{The committee $C$ can now use the destination $\cml$'s of all surviving tokens with source committee $C$ as random samples.}
		\normalsize 
	\end{algorithmic}
\end{algorithm}
\begin{theorem}\label{thm:sampling}
	The sampling procedure described in Algorithm~\ref{alg:sampling} where each committee initiates $\Theta(\alpha \log^3 n)$ tokens conforms to the model described in Section~\ref{sec:model}.  Furthermore, the following two statements hold.\\
{\bf 1.} For any token $t$ (chosen when tokens were first generated), if $t$ survived all the iterations  in Algorithm~\ref{alg:sampling}, then the final destination  is equally likely to be any of the $N$ committees in Spartan.  \\
{\bf 2.}  With high probability, for every committee $C$, the number of tokens with $C$ as the source and the number of tokens with $C$ as the destination  are both  at least $\Omega(\alpha \log n)$.
	\end{theorem}
\begin{proof}

	We first prove statement 1 followed by statement 2. The correctness and conformity to the model specifications follow suit.
	
	To prove statement 1, we first recall 
	a basic fact about butterfly networks. Consider a walker in possession of the binary representation of some $\r^* \in [1, 2^k]$. Let us suppose she performs the following bit-fixing moves $k$ times from her current position $(\r, \c)$. For any $i=0\ldots k$ bit fixing moves, let $j=((i+\c)\mod n)+1$ (using the decimal representation of $\c$). At any step $i$ in the bit fixing process, if the $j$th bit of $\r$ and $\r*$ are equal, she moves along the direct edge; otherwise, she moves along the flip edge.

	 Using the above procedure, at the end of $k$ moves, she will be in row $\r^*$~\cite{MU05}. Thus, when we perform  the same bit-fixing procedure with a random $k$ bit vector instead, we will reach a random row. It now suffices to show that any token $f$ that survived till the end of the first \textbf{for} loop in Algorithm~\ref{alg:sampling} performed the required $k$ random bit fixing steps. (The random column is obvious from the second \textbf{for} loop.) During round 1, clearly every token  has taken a random bit-fixing step, i.e., either along the direct edge or the flip edge. For the sake of induction, let us now suppose  that at the end of round $r$, every token  has made $2^{r-1}$ random bit fixing steps; we have already established the basis for this. During round $r+1$, $f$  is matched with another token, say $s''$, both of which have fixed $2^{r-1}$ bits (see Figure~\ref{fig:tokens}). Thus, when $f$ is matched with $s''$,   $f$ inherits the bit-fixing steps performed by $s''$, and therefore, the updated $f$ has fixed $2^r$ bits from its source column by the end of round $r+1$. 
	
	\begin{figure}[htbp]
  \begin{center}
    \includegraphics[width=0.4\textwidth,clip=true,trim=50 182 70 162]{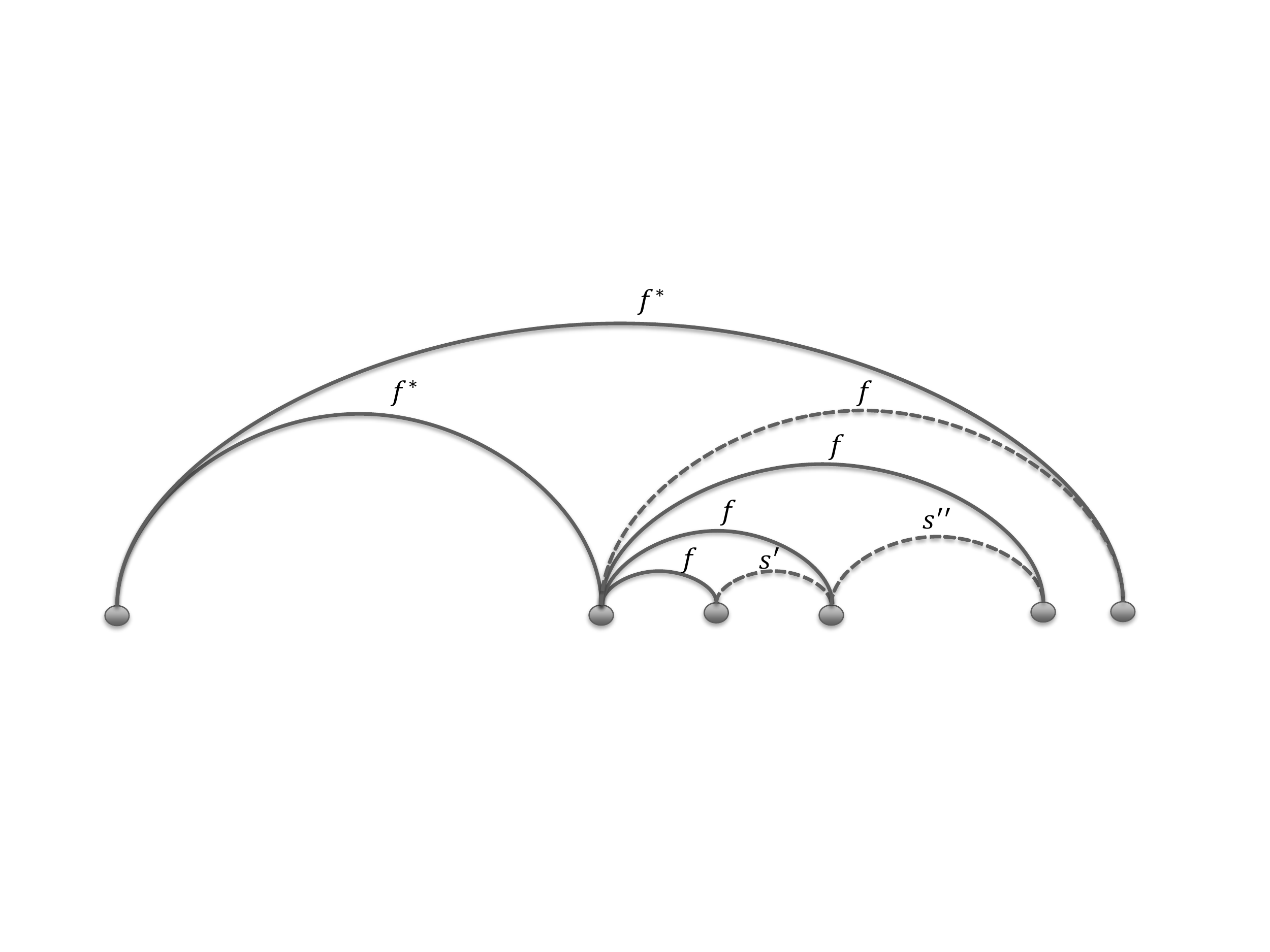}
  \end{center}
  \caption{Schematic showing path doubling of a token $f$.  The tokens shown in solid lines play the role of the first token in the matching of tokens, while those shown in dotted lines are the second. As long as $f$ is the first token, it is retained for subsequent doubling, but once it is used as a second token (with $f'$ as the first token in the figure), it is discarded.}\label{fig:tokens}
\end{figure}

We now turn our attention to statement 2. Note that, in every round $r$ for every $C \in \C$, at the start of every round, each committee has the same value of $L$, as each of them started with the same number of tokens and the value of $L$ is the same.

\begin{claim} \label{clm:bounds}
The following bounds hold with high probability at the start of every round $r$. For any fixed $\epsilon >0$\\
1. The number of tokens with $C$ as  source and $rand \ge r$ is within $[(1 - \epsilon) L /2, (1+ \epsilon) L /2]$.\\
2. The number of tokens with $C$ as  source and $rand = r-1$ is within $[(1 - \epsilon) L /2, (1+ \epsilon) L /2]$.\\
3. The number of tokens with $C$ as  destination is within  $[(1 - \epsilon) L, (1+ \epsilon) L]$.\\
4. The number of tokens with $C$ as  destination and $rand \ge r$ is within $[(1 - \epsilon) L /2, (1+ \epsilon) L /2]$.\\
5. The number of tokens with $C$ as  destination and $rand = r-1$ is within $[(1 - \epsilon) L /2, (1+ \epsilon) L /2]$.
\end{claim}

\begin{proof}
 The first two bounds are straightforward applications of Chernoff bounds because each of $L$ tokens with $C$ as source is independently and equally likely to have $rand \ge r$ or $rand = r-1$; recall that the rand values follow the geometric distribution with parameter 1/2 and all tokens with $rand<r-1$ have already been discarded. Once bound 3 is proved, bounds 4 and 5 will also follow in a fashion similar to 1 and  2. So we will now focus our efforts on bound 3. 

At the start  of round $r$, we know that the following two statements are true of any surviving token i) it has walked for $2^{r-2}$ steps. ii) Because of the inherited bit fixing in the first \textbf{for} loop of the algorithm \ref{alg:sampling}, the destination of a surviving token is equally likely to be any of $2^{2^{r-2}}$ committees. Thus by symmetry, if you fix a committee $C$, the possibility of $C$ being a token's destination is $1/2^{2^{r-2}}$. Thus, since each committee in a round $r$ would be propagating at least $L$ tokens (Note that since $L$ is constantly updated due to step 7, it is thus is a lower bound for both $|F|$ or $|S|$), there are thus a total of at least $L2^{2^{r-2}}$ tokens. Which means that on expectation at least $L$ tokens will have their destination as $C$, we may then use standard Chernoff bounds~\cite{MU05} to prove bound 3. Bounds 4 and 5 follow in the same manner as bounds  1 and 2. Tokens with $C$ as destination are equally likely to have $rand\geq r$ or $rand=r-1$. Thus, following bound 3, we can use Cheronoff bounds to show that bounds 4 and 5 follow. 
\end{proof}
The rest of the proof is conditioned on the bounds stated in Claim~\ref{clm:bounds}. To complete the proof of the theorem, we need to show for  every committee $C$ there will (whp)  be  at least $L/4$ tokens that have $C$ as source at the end of one iteration of the algorithm. Once this is shown, we know that, after $O(\log \log n)$ rounds, the number of tokens will (whp) be  at least $\Omega(\frac{\alpha \log^3 n}{4^{\log \log n}}) = \Omega(\alpha \log n)$, which will complete the proof.

Let's look at the very last round $r$ in the first \textbf{for} loop, that is the very last matching of tokens. Let $\Psi$ be the set of tokens with $C$ as source  and $rand \ge r$; we know $|\Psi| \ge (1 - \epsilon) L /2$ (holds true due to the inductive nature of Claim 13.1). For any token $t \in \Psi$, the destination in round 
$r$ is going to be a committee $C'$, which had been set by inheriting the bit fixing steps of a token $t'$ in the previous round. Now at $C'$, any token $t$ that is not matched with some $t''$ is going to be discarded. What is the probability of $t$'s survival? We know that there can be at most $(1+\epsilon)L/2$ tokens at $C'$ whose destination is $C'$ and we also know that there at least $(1-\epsilon)L/2$ tokens at $C'$ whose source is $C'$ and whose random number is $r-1$ (because of claim 13). Thus for $t$ to not survive, it should be discarded at $C'$, the probability of which is exactly $\delta=(1-\epsilon)/(1+\epsilon)$. Notice that the dependencies are negatively correlated, i.e., when we know that some token is discarded (at $C'$) it only decreases the probability that some other token is discarded (as it increases the probability that the second token is matched). To see how, we may again look at the token matching phenomenon as a balls and bins scenario, in which the tokens that arrive at $C'$ are the balls and the tokens whose source is $C'$ are the bins. Again using the properties of \ref{lem:bbm} and \ref{lem:dma}, we can argue negative correlation. Thus after the the last matching, the probability that a surviving token $t \in \Psi$ is not discarded is at least $\frac{2\epsilon}{(1+\epsilon)}$. We can then calculate the expected number of tokens that survive from $\Psi$ as at least $\frac{2\epsilon}{1+\epsilon}\frac{(1-\epsilon)}{2}L=\frac{\epsilon(1-\epsilon)}{1+\epsilon}L$. Thus, by applying the generalized Chernoff bounds from~\cite{MU05}, 
\begin{align*}
\mathbf{Pr}(\text{More than $L/4$ tokens are discared}) &\leq \mathbf{exp}\left\lbrace-\frac{\frac{\epsilon(1-\epsilon)}{1+\epsilon}L {\left(\frac{(1+\epsilon)(5\epsilon-1)^2}{4\epsilon(1-\epsilon)}\right)}^2}{4} \right\rbrace
 \\
&\leq \textbf{exp}\left\lbrace\frac{-(\epsilon+1){(5\epsilon -1)}^2}{ (1-\epsilon^2)}L/16\right\rbrace
\end{align*}

we can prove that the probability that more than a fourth of tokens will be discarded is at most $\textbf{exp}\left\lbrace\frac{-(\epsilon+1){(5\epsilon -1)}^2}{1-\epsilon^2}L/16\right\rbrace$. Thus with high probability,  at least $L/4$ tokens survive. Any such surviving token then inherits a final bit fixing step due to step 8 and since there are no possibilities of discarding the tokens thereafter, all of the tokens that have survived steps 1 through 7 end up at $C$ with a destination chosen uar from the butterfly network. 
This also concludes the proof of statement \ref{thm:sampling}, thus proving that each committee can have $O(\log n)$ tokens at the end of Algorithm \ref{alg:sampling}. 
\end{proof}

During the maintenance phase, this sampling procedure must be initiated repeatedly every $\Theta(\log \log n)$  rounds with $\alpha = \log \log n$. Thus, each committee  gets $\Theta(\log n \log \log n)$ tokens every $\Theta(\log \log n)$  rounds and they can be distributed among the $\Theta(\log n)$ nodes in the committee.
There is one subtlety that must be resolved. How do the nodes within a committee share the  tokens?  We assume that the list of members in the committee is common knowledge among all the members. The node with the smallest ID then divides up the tokens into chunks of $\Theta(\log \log n)$ tokens and sends each chunk to a member so that every member has at least one chunk.

\emph{Remark}: This random walks sampling procedure may require nodes to send/receive $\omega(\log n)$ messages (but not more than $O(\plog(n))$ per round). We believe this is an inherent bottleneck that we cannot avoid when each committee requires $\Theta(\log n)$ random walks within a period of $O(\log \log n)$ rounds. To see this, notice that the total number of random walk samples needed is $O(n\log n)$ and each of them -- to mix properly -- must have walked $\Omega(\log n)$ steps. This means that the network must send/receive at least $\Omega(n\log^2 n)$ messages within a period of $O(\log \log n)$ rounds, which, by the pigeonhole principle, will require nodes to send/receive $\omega(\log n)$ messages per round. This however does not preclude the possibility of a more complicated sampling algorithm that may work with no more than $O(\log n)$ messages per node per round.

\subsection{Maintaining the  Spartan Implementation}\label{subsec:maintain}
We now turn our attention to maintaining Spartan. At the end of the bootstrap phase, our network has $N = k2^k \in \Theta(n/\log n)$ committees. Our goal in the maintenance phase is to ensure  that all such committees and logical edges are robust (whp) for $\poly(n)$ number of rounds despite a churn of up to $(\epsilon n, \Theta(\log \log n))$. 

 {\bf Sampling Cycles.} Throughout the maintenance phase for as long as Spartan is robust, all the committees repeatedly execute the random walks based sampling procedure in unison once every $\Theta(\log \log n)$ rounds. They set $\alpha = \Theta(\log \log n)$. Consequently, every committee will obtain $\Theta(\log n \log \log n)$ samples every $\Theta(\log \log n)$ rounds. These samples will remain valid for the next $\Theta(\log \log n)$ rounds during which they will be consumed. These sampling cycles are illustrated in Figure~\ref{fig:sampling-cycle}.

\begin{figure}[htbp]
	\begin{center}
		\includegraphics[scale=0.45,clip=true,trim=60 240 0 210]{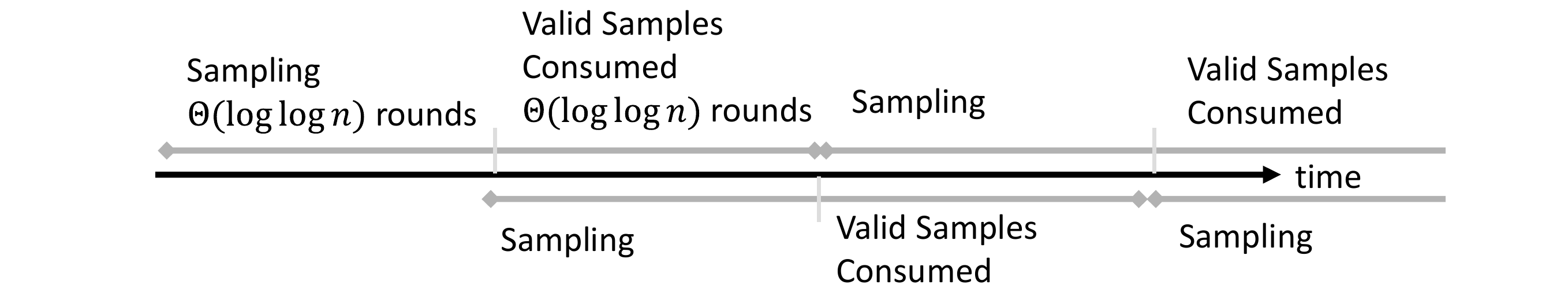}
		\caption{Sampling cycles run in the background for as long as all the committees are robust.} \label{fig:sampling-cycle}
	\end{center}
\end{figure}

With the sampling cycles executing repeatedly in the background, the nodes perform their maintenance protocol (see Algorithm~\ref{alg:maintainence} for an event-driven pseudocode). At a high-level, the nodes move to new random committees every few sampling cycles. 

{\bf Node-level Invariants.} Throughout the maintenance phase, we need to maintain two crucial invariants (mentioned below) for every node in the network (after its first $\Theta(1)$ rounds). At a high-level, the sampling cycles ensure that the first invariant is guaranteed and the first invariant ensures that the nodes can maintain the second invariant. The second invariant then ensures that all the committees (and therefore Spartan as a whole) will be robust (whp) against a $\Theta(\log \log n)$-late adversary, which in turn ensures that sampling will work, thereby completing the cycle of dependencies.  See Figure~\ref{fig:invariants} for an illustration.

\begin{figure}[htbp]
	\begin{center}
		\includegraphics[scale=0.6,clip=true,trim=60 240 35 190]{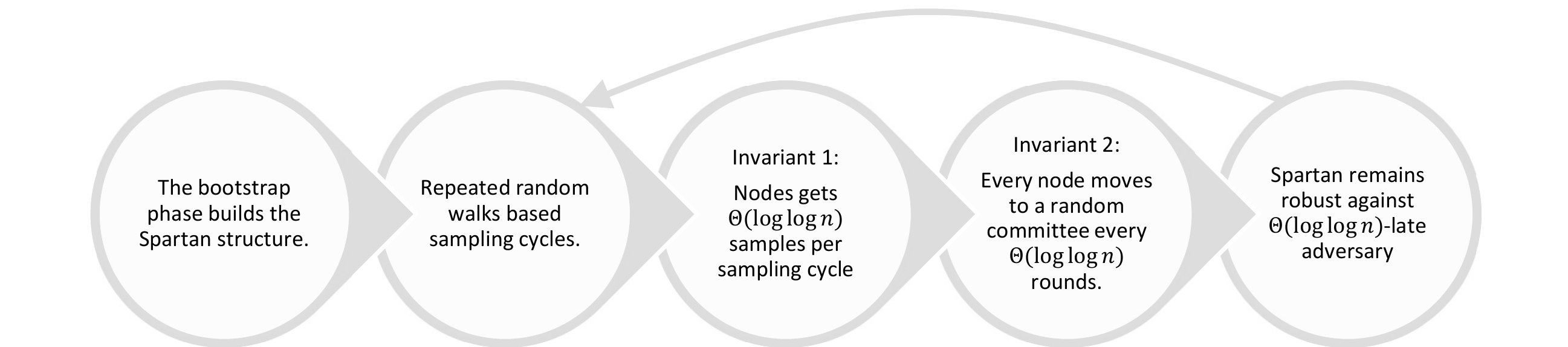}
		\caption{Interaction between sampling, the node-level invariants, and robustness of Spartan.}
		\label{fig:invariants}
		
	\end{center}
\end{figure}

\begin{description}
\item[Invariant 1: Random Samples.] Every node  has at least two valid CMLs corresponding to random committees at all times. 

We will now briefly describe the behavior of a node in the Spartan maintenance cycle (Algorithm~\ref{alg:maintainence} has more details). When a new node $u$ enters the network, it is seeded with the ID of an arbitrary node $v$. Node $u$ then requests $v$ and receives a CML corresponding to a random committee that it uses to move to a random committee $C$. As a new member of $C$, it receives $\Omega(\log \log n)$ valid CMLs and this will satisfy $u$'s requirement until the next sampling cycle. 

During each sampling cycle, each committee $C$ obtains  $\Theta(\log n \log \log n)$ CMLs. They  are apportioned amongst the $O(\log n)$ members such that each member gets $\Theta(\log \log n)$ samples. Moreover, a set of $\Theta(\log n \log \log n)$ samples are set aside for the up to $O(\log n)$ new nodes that may become members of $C$. 
\item[ Invariant 2: Random Committee Membership.] Every node $u$ is in a committee chosen uniformly and independently at random. Recall that $u$ receives a CML pertaining to a random committee $C$ in its first round. It then moves to $C$ in its second round. Subsequently, every few sampling cycles, it moves to a new (uniformly and independently chosen) random committee. 
\end{description}

Each new node $u$ that enters the network will in 2 rounds (whp)  become  a member  of some committee. Recall that the new node $u$ will be seeded with the ID of an arbitrary pre-existing node $v$ in the network. In the first (handshake) round, $u$ will request $v$ and obtain  $\cml(C)$ pertaining to a random committee $C$ from the node that it connected to (see Section~\ref{sec:model} for more detail). In the second round, it will connect to $C$ using $\cml(C)$.  

The node $u$ will then operate in cycles of $\Theta(\log  \log n)$ rounds (on expectation) by moving to a new random committee every few (constant number of) sampling cycles. This process ensures that $u$ becomes a member of a  random committee periodically and also ensures that the committees are periodically refreshed. This ensures that the $\Theta(\log \log n)$-late adversary is incapable of inferring the members within committees.
\begin{algorithm}
\caption{Node level protocol for maintaining Spartan.}\label{alg:maintainence}
\begin{algorithmic}[1]
\small
    \item[{\bf Note: }] For clarity, the protocol is described in an event-driven manner from the perspective of a single arbitrary node $u$.

	\item[{\bf New Arrival: }] When a new node $u$ enters the network, it uses its seed knowledge (ID of a pre-existing node $v$) to request and obtain the CML of a random committee $C$. Node $u$ then requests membership and joins $C$.
	\item[{\bf Sampling: }] Node $u$ participates in the sampling procedure that is repeatedly performed. Node $u$ then obtains $\Theta(\log \log n)$ valid samples that will sustain its requirements until the next sampling completes. 
	\item[{\bf Helping a new arrival: }] Whenever a new node $w$ requests, $u$ provides $w$ with the CML of a random committee.  
	
	\item[{\bf Moving to a new committee: }] \label{lno:moving} At the start of each sampling cycle, node $u$  moves to a new committee with some small constant probability $p \ll 1/2$ and stays put with probability $1-p$. The parameter $p$  can be adjusted for convenience. For example, small values of $p$ will limit the nodes from moving too frequently.

	The actual act of moving is straightforward. Suppose $u$ wishes to move to a committee $C$, it sends a ``join-request" message to each node in the $\cml(C)$ and those that are currently in $C$ (which will be at least $\Omega(\log n)$ in number) will accept $u$ into $C$.  (Node $u$ must also drop all its edges with its previous committee.) 
	
	\item[{\bf The node $u$ has been in a committee for too long.}] We do not want any node to stay in the same committee consecutively for more than a constant number of sampling cycles. Otherwise, the adversary may be able to infer some of the members within a committee based on its knowledge of the committee from $\Omega(\log \log n)$ rounds ago. Therefore, the node must override the probabilistic move and deterministically move if it has been in a committee for more than a sufficiently large constant (say, 10) number of sampling cycles. 
	
	\item[{\bf A fellow committee member leaves: }] Members of a committee may leave (without notice) either because they moved to a new committee or were just churned out. In either case, whenever node $u$ detects that a fellow member has left its committee, $u$ must update the local committee members list and also inform neighboring committees of the update. 
\item[{\bf Accepting new members: }] To facilitate the movement of a node $v$ to another committee , if $u$ receives a ``join-request" message from $v$ that is currently not in $u$'s committee, then $u$ must honor it and include $v$ into its committee. This may entail updating the local committee members list and informing neighboring committees periodically (at least once every $O(\log \log n)$ rounds) about the committee's updated member list.

\normalsize
\end{algorithmic}
\end{algorithm}

{\bf Analysis Overview.} We now wish to show that the maintenance procedure described above will work (i.e., invariants maintained) and that Spartan will remain robust for some sufficiently large $\poly(n)$ rounds. 

We begin by recalling a well-known claim with regards to the balls-into-bins model (see \cite{MU05,DP09}).
\begin{claim}\label{clm:bib}
Consider $m$ bins. For any fixed $c$, there exists a constant $c'$ such that when $c' m \log m$ balls are thrown uniformly and independently at random into $m$ bins, then, every bin will contain $\Theta(\log n)$ balls with probability at least $1 - n^{-c}$.  
\end{claim}
In our context, the committees are the bins and nodes are the balls and the proof works by showing that every committee will contain $\Theta(\log n)$ nodes for a sufficiently long time. 

Formally, our analysis works by induction. Let us use $\ell = (1, 2, \ldots)$ to index the sampling cycles sequentially.  The first sampling cycle starts right after the bootstrap phase, so clearly Spartan is robust until the start of the first sampling cycle. Moreover, by the construction of Spartan, each node is in a committee that was chosen uniformly and independently at random. Thus, by immediate application of Claim~\ref{clm:bib}, Spartan is robust at the start of sampling cycle 1. This establishes our basis. 

\begin{lemma} \label{lem:robust}
If  Spartan is robust at the start of the $\ell$th sampling cycle, then, it will (whp) remain robust until the end of the cycle (equivalently, the start of $(\ell+1)$th cycle). Moreover, each node is placed in a committee that is chosen uniformly at random and independent of all other nodes' committee assignments. 
\end{lemma}
\begin{proof}

First, we note  that within sampling cycle $\ell$, no committee will become empty. Since the cycle duration is $\Theta(\log \log n)$ rounds, at most $\epsilon n$ nodes would have churned out. We must importantly note that the $\Theta(\log \log n)$-late adversary will be unaware of where the $\epsilon n$ nodes are, so the choice of nodes that are churned out are essentially uniformly at random. To formalize this, we can use the principle of deferred decisions. For any choice of $\epsilon n$ nodes, we can defer the choice of their committee to the round when they are churned out.  Consequently, the remaining nodes (of cardinality at least $n - \epsilon n \in \Omega(n)$) is also a uniformly random subset and sufficiently large, thereby allowing  us to apply Claim~\ref{clm:bib}. We can in fact claim that every committee will have at least $\Omega(\log n)$ nodes (whp). 

Secondly, every new node that entered the network during sampling cycle $\ell$ was placed in a random committee. Since the number of such new nodes is at most $\epsilon n$, again, by Claim~\ref{clm:bib}, we get the number of nodes within each committee to be $O(\log n)$ whp. 
\end{proof}

It is largely clear that all the rules of the model are followed during our construction. We now point out a few details that may not be immediately obvious. Firstly, each node forms overlay edges with all members of its current committee as well the members of the three neighboring committees. No other overlay links are formed. This means that no node violates the $\Delta \in O(\log n)$ degree restriction.  From the description, it is clear that we are never sending more than $O(\plog(n))$ messages of size at most $O(\plog(n))$ bits each. In addition, no node receives more than $O(\plog(n))$ messages either (whp) and this is clarified in several places. The general principle is that messages are either sent across the overlay links (thus obviously ensuring that  a node can receive at most $O(\log n)$ of such overlay messages) or messages are sent to random nodes via random sampling. Such messages sent to random nodes -- by applying the balls-into-bins claim (see Claim~\ref{clm:bib}) -- cannot disproportionately land on any one node. Thus, we can conclude with the following theorem. 

\begin{theorem} \label{thm:main}
Despite heavy churn at the rate of $[\epsilon n, O(\log \log n)]$ controlled by a $\Theta(\log \log n)$-late adversary whose actions are instantaneous, we  can, whp, bootstrap and maintain the Spartan overlay network without violating the rules of the model laid out in Section~\ref{sec:model}. Moreover, the Spartan network so constructed and maintained has the following properties. 
 \begin{enumerate}
	\item Spartan provides $O(n/\log n)$ addressable locations called committees arranged in a wrapped butterfly network.
	\item The diameter is $O(\log n)$ allowing any node in the network to reach another node in $O(\log n)$ rounds.
	\item Spartan is guaranteed to be robust for a time period that is an arbitrarily large $\poly(n)$.
\end{enumerate}
\end{theorem} 
\section{Experimental analysis}

While our theoretical results are largely insightful, for the sake of simplicity, we have not worked out all the constants. However, using the correct numbers is important when making key design choices like the size of the butterfly network (i.e., the number of committees N) required for robustness. Put another way, for a given number of committees $N = k2^k$, what would be the minimum number of nodes that we need in the network in order to ensure that all committees are populated with member nodes for a sufficiently long period of time?  Essentially, we have to ensure that the committees are, on average, assigned sufficiently many peers. The theory we have developed indicates $\Theta(\log n)$ peers per committee on average, which, assuming a conservative estimate of 10 for the constant and $n$ being 1024 (again, quite conservatively), will be $\approx 100$ nodes per committee. We believe this is a significant overestimate; smaller committees should suffice.  To obtain a more realistic estimate of the average number of nodes per committee, we performed some simple experimental analysis. 

We consider $n$ peers assigned randomly to  $N=k2^k$ committees at round 1 to simulate the idea of peers being distributed among committees. We define a committee to be robust up to round $r$ if it has been non-empty until round $r$. A bad event in this setting is when at least one committee becomes empty. To simulate churn, in every round after $r=1$, $\epsilon n$ peers (chosen UAR) are removed and $\epsilon n$ are placed  randomly into the committees. We repeat the above procedure until we either reach a previously defined stopping point of $R$ rounds (we use $R=10000$ rounds) or until we reach a round in which at least one committee is empty. If all committees remain robust for $R$ rounds, we declare that the repetition was successful.

We considered $N = k2^k$ committees (varying $k$ in the range $[5,10]$) and $\epsilon=.1$. For each value of $k$, we ran  $30$ repetitions of the experiment. For a given $k$, we define the \emph{threshold} $T$ to be the minimum number of peers  such that all $N=k2^k$ committees are robust through at least 90\% of all repetitions (i.e., at least 27 repetitions). Then from Table~\ref{tab:threshold}, we see that the number of successful repetitions drastically decreases even for $n= 0.9 \times T$. Thus, we need to be careful in choosing the number of committees.

\begin{table}[htpb]
\begin{tabular}{|l|l|l|l|l|}
\hline
\multirow{2}{*}{Number of committees $N=k2^k$} & \multirow{2}{*}{ Threshold $T$} 
& \multicolumn{3}{c|}{Number of failed repetitions} \\ \cline{3-5}
& & $T$  & $0.9T$   &  $0.8T$  \\ \hline
160 & 2880           & 0 & 10 & 28 \\ \hline
384  & 7680           & 0 & 10 & 27 \\ \hline
896    & 17920           & 0 & 11 & 30 \\ \hline
2048  & 40960       & 3 & 21 & 30 \\ \hline
4608  & 100000          & 3 & 18 & 30 \\ \hline
10240 & 250000         & 0 & 9  & 30 \\ \hline
\end{tabular}

\caption{$N=k2^k$ represents the number of committees. For each value of $N$, the Threshold ($T$) represents the number of peers such that 27 out of 30 repetitions were successful. Subsequent columns represent number of failed repetitions out of 30 for $n=T$, $n = .9T$, and $n = .8T$, respectively. As we can see, even for $n= 0.9 \times T$, there is a drastic decrease in the number of successful repetitions. }
\label{tab:threshold}
\end{table}
The analysis also allows one to discover the value of the threshold $T$ for different values of $N$. Let $c$ be the multiplicative factor such that threshold $T=cN$  for various values of $N$; this ensures that the average number of peers per committee is $c$.  From Figure~\ref{fig:constant} we observe that as $T$ increases, $c$ also increases, which is of course to be expected. However, one can immediately notice that the factor $c$ is in general much smaller than the size of the committees required for our theoretical analysis to go through.   
\begin{figure}[htbp]
	\begin{center}
		\includegraphics[scale=0.4]{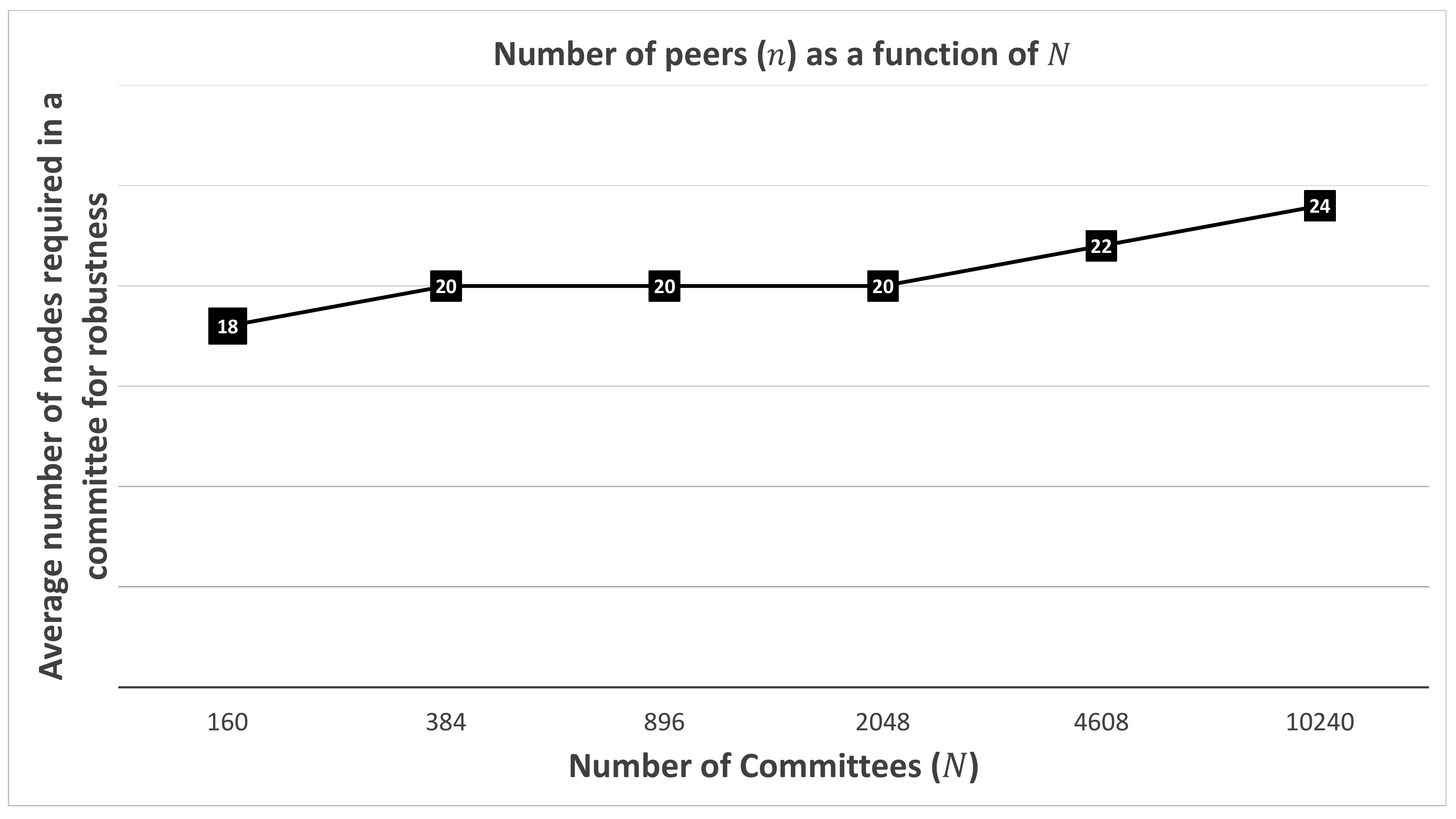}
		\caption{The average number of peers required for each committee in $N$ (for $5\leq k\leq10$) to guarantee robustness.}
		\label{fig:constant}
		
	\end{center}
\end{figure}

\section{Discussion and Concluding Remarks}
\label{sec:conc}
We have presented Spartan, a sparse overlay network that is robust against heavy churn designed by an adversary that is aware of the network except for the most recent  $O(\log \log n)$ rounds. The overlay provides us with $\Theta(n/\log n)$ addressable committees in which data items can be stored and maintained in a robust manner. Furthermore, each committee is easily accessible to all other committees (within $O(\log n)$ hops) because they are arranged in the form of a wrapped butterfly network. 

We believe that this basic framework is quite flexible and can be adapted and  extended in a variety of interesting ways to tackle some nagging issues in overlay network design. 

For a start, in our model, we have assumed that the number of nodes must lie in $[n, fn]$ for some $f \ge 1$. This is actually not a strict requirement. There are established techniques to count the number of nodes in the network~\cite{APRU12,AKS15}. When the count gets too close to either end of the range $[n, fn]$, we can appropriately halve $n$ or double it as required. This would mean that other parameters like $N$ and $k$ (the number of columns) must also be recalculated.  

We also have the flexibility to  string the committees together to form other structures like hypercubes, expander graph or de Bruijn graph.   Apart from using just a single type of graph, we could also design overlays that combine for example both a butterfly and a de Bruijn graph. Since both are low  degree graphs, we can inherit the benefits of both families of graphs without affecting our results.

In our model, we have to do a lot of work in order to keep the network ``warmed up" for all possible adversarial strategies. However, churn characteristics have been well studied and in particular, it is well known that despite heavy churn there will be a significant fraction of stable users who churn at a much lower rate~\cite{SR06}. We believe that the model could be enhanced to include a more careful accounting of the amount of work that is performed overall, which could in turn lead to more careful algorithm design that is more work efficient. In particular, in this context, it will be interesting to see if recent ideas of resource competitiveness surveyed by Bender et al.~\cite{ResComp15} in which resource expenditures are weighed against the adversary's efforts could lead to more nuanced algorithms more amenable to real-world implementation. 

Finally, Spartan is built with a wide range of applications in mind. It is a natural candidate to implement distributed hash tables capable of storing $\langle$key, value$\rangle$ pairs. The keys can easily be hashed into the space of committee identifiers. It will also be interesting to see if erasure codes can be used effectively in tandem with the idea of committees. Another feature is that Spartan is naturally resistant to attacks like censorship and denial-of-service. If a subset of the nodes is either censored or bombarded by denial-of-service attacks, the rest of the nodes will still be able to keep the committees operating. 

For future work, we wish to explore the possibility of reducing the degree down to a constant. This seems plausible given recent works like~\cite{APRRU15} where constant degree expanders were maintained despite churn, but it is not clear as yet whether a constant degree addressable and routable \Pe network is possible.
\section*{Acknowledgment}

    The authors would like to thank  Chetan Gupta, Christian Scheideler, and Eli Upfal for  useful discussions and ideas. We thank the anonymous reviewers for useful suggestions.

\bibliographystyle{unsrt}  




\end{document}